\documentclass[sigconf]{acmart} % add review to show line numbers
\usepackage{color}	
\usepackage{xspace}
\usepackage[utf8]{inputenc}
\usepackage{wrapfig}

\usepackage{colortbl} % add in the preamble if not already
\usepackage{xcolor}   % for custom colors
\usepackage{graphicx}
\usepackage{color}
\usepackage{enumerate}
\usepackage{hyperref}
\usepackage[font={small,sf}]{caption}
\usepackage{makecell}
\usepackage{floatflt}
\usepackage{lipsum}
\usepackage{multirow}
\usepackage{soul}
\usepackage{mdframed}
\usepackage{amsmath,amsfonts}
\usepackage{comment}
\usepackage{subcaption}
\usepackage{enumitem}
\usepackage{csquotes}
\usepackage{tabularx}
\usepackage{tabularx}
\usepackage{enumitem}
\usepackage{tikz}
\usetikzlibrary{positioning}
\usepackage{mathtools}
\usepackage[a-1b]{pdfx}  % load first if possible
\usepackage[T1]{fontenc}
\usepackage{microtype}
\usepackage{bbm}
\usepackage{fnpct}  % adjusts kerning of footnote marks (superscripts) before commas or periods

\graphicspath{{./figs/}}

\usepackage{algorithm}
\usepackage[noend]{algpseudocode}
\DeclarePairedDelimiterX\set[1]\lbrace\rbrace{\,#1\,}
\usepackage{tikz}
\usetikzlibrary{positioning}
\newcommand{\stitle}[1]{\vspace{1mm}\noindent{\textbf{#1}}.}
\newcommand{\eps}{\varepsilon}

\newcommand{\enet}{{$\varepsilon$-{net}}\xspace}
\newcommand{\enets}{{$\varepsilon$-{nets}}\xspace}
\newcommand{\esample}{{$\varepsilon$-{sample}}\xspace}

\newcommand{\at}[1]{{\tt \small #1}\xspace}

\newtheorem{theorem}{Theorem} 
\newtheorem{lemma}[theorem]{Lemma}
\newtheorem{definition}{Definition}
\newtheorem{example}{Example}
\newcommand{\mohsen}[1]{\textcolor{red}{Mohsen: #1}}

\newcommand{\ranges}{\mathcal{R}}
\newcommand{\points}{X}
\newcommand{\rangespace}{(\points, \ranges)}
\newcommand{\range}{R}
\newcommand{\point}{p}
\newcommand{\epsnet}{\mathcal{N}}
\newcommand{\epssample}{\mathcal{A}}
\newcommand{\group}{c}
\newcommand{\groups}{C}

\DeclareMathOperator{\VC}{VC}
\newcommand{\opt}{\mathsf{OPT}}

\renewcommand\footnotetextcopyrightpermission[1]{}
\AtBeginDocument{%
  }

\renewcommand{\Re}{\mathbb{R}}
\newcommand{\prob}{\varphi}
\newcommand{\sizenet}{\lambda}
\newcommand{\parti}{\mathcal{P}}
\renewcommand{\root}{\mathsf{root}}
\newcommand{\X}{\mathcal{X}}
\newcommand{\res}{\mathcal{S}}

\begin{document}
\title{On Fair Epsilon Net and Geometric Hitting Set}

\author{Mohsen Dehghankar}
% \orcid{XXX}
\affiliation{%
  \institution{University of Illinois Chicago}
  \city{Chicago}
  \state{IL}
  \country{USA}
}
\email{mdehgh2@uic.edu}

\author{Stavros Sintos}
\orcid{0000-0002-2114-8886}
\affiliation{%
  \institution{University of Illinois Chicago}
  \city{Chicago}
  \state{IL}
  \country{USA}
}
\email{stavros@uic.edu}

\author{Abolfazl Asudeh}
\orcid{0000-0002-5251-6186}
\affiliation{%
  \institution{University of Illinois Chicago}
  \city{Chicago}
  \state{IL}
  \country{USA}
}
\email{asudeh@uic.edu}

\acmConference{}{}{} % Clears conference name and location
\acmYear{}           % Clears the year
\acmDOI{}            % Clears DOI
\acmISBN{}           % Clears ISBN

% \settopmatter{printacmref=false}    % Removes citation info below the title


\begin{abstract}
Fairness has emerged as a formidable challenge in data-driven decisions. Many of the data problems, such as creating compact data summaries for approximate query processing, can be effectively tackled using concepts from computational geometry, such as $\eps$-nets. However, these powerful tools have yet to be examined from the perspective of fairness.

To fill this research gap, we add fairness to classical geometric approximation problems of $\eps$-net, $\eps$-sample, and geometric hitting set. We introduce and address two notions of group fairness: demographic parity, which requires preserving group proportions from the input distribution, and custom-ratios fairness, which demands satisfying arbitrary target ratios.

We develop two algorithms to enforce fairness—one based on sampling and another on discrepancy theory. The sampling-based algorithm is faster and computes a fair $\eps$-net of size which is only larger by a $\log(k)$ factor compared to the standard (unfair) $\eps$-net, where $k$ is the number of demographic groups. The discrepancy-based algorithm is slightly slower (for bounded VC dimension), but it computes a smaller fair $\eps$-net. Notably, we reduce the fair geometric hitting set problem to finding fair $\eps$-nets. This results in a $O(\log \mathsf{OPT} \times \log k)$ approximation of a fair geometric hitting set.

Additionally, we show that under certain input distributions, constructing fair $\eps$-samples can be infeasible, highlighting limitations in fair sampling. 
Beyond the theoretical guarantees, our experimental results validate the practical effectiveness of the proposed algorithms. In particular, we achieve zero unfairness with only a modest increase in output size compared to the unfair setting.
\end{abstract}

\maketitle

\vspace{-2mm}
\section{Introduction}

\subsection{Motivation}
As algorithmic decisions continue to influence critical aspects of modern human life, from resource allocation and recommendation to hiring and even predictive policing, the need to ensure fairness in data-driven systems has become increasingly urgent. 

Geometric approximation algorithms and notions such as \enets facilitate {\em Approximate Query Processing} ({\sc AQP}) by providing compact data summaries, aka data representations, that preserve key properties of large datasets~\cite{mustafa2017epsilon, haussler1986epsilon}. These geometric tools are especially valuable due to their provable guarantees of approximation quality and computational efficiency. 
\enets, for example, are well-known as small subsets of data that guarantee to contain at least one sample from each range larger than $\eps$, and hence can be used for fast query answering.
To further clarify this, let us consider the following example for approximate database range-query processing.

\begin{figure}[!thb]
    \centering
    \begin{tikzpicture}
\def\blues{(-2,1.5), (-1,1), (-2,3), (1,2.5), (1,3),(-.2,2), (-.6,1.7),(0.1,2.9),(-0.4,3)}
\def\reds{(-2.1,2.5), (-1.5,1), (1.4,2), (0,3.5), (1.3,2.8),(.2,1), (-.3,1.1), (-1,1.4),(-0.6,3.1)}
\def\netone{(-2.1,2.5), (1.3,2.8),(0,3.5), (-1,1.4),(-.2,2)}
\foreach \Point in \blues{
    \node[color=blue!60] at \Point {\textbullet};
}
\foreach \Point in \reds{
    \node[color=red!60] at \Point {\textbullet};
}
\foreach \Point in \netone{
    \draw [color=green!100] \Point circle [radius=0.15];
}

% rectangles
\draw [draw=black] (-0.8,3.3) rectangle (1.5,2.65);
\node at (1.6,3.4) {$r_1$};

\draw [draw=black] (-1.2,0.5) rectangle (.5,2.2);
\node at (.7,1.8) {$r_2$};

% x-axis and y-axis
\draw[->, thick] (-2.7,0) -- (2.2,0) node[right] {$x$};
\foreach \x in {-2, -1, 0, 1, 2}
    \draw (\x,0.1) -- (\x,-0.1) node[below] {\x};

\draw[->, thick] (-2.5,-0.2) -- (-2.5,4.2) node[right] {$y$};
\foreach \y in {1,2,3,4} {
    \draw (-2.6,\y) -- (-2.4,\y);
    \node at (-2.75,\y) {\y};
}

\end{tikzpicture}

% \begin{tikzpicture}
% \draw[help lines, color=gray!30, dashed] (-1.2,0) grid (2.5,4);
% \draw[->,ultra thick] (-1.2,0)--(2.5,0) node[below]{$x$};
% \draw[->,ultra thick] (0,0)--(0,4) node[right]{$y$};
% \draw [fill] (-.8,0) circle [radius=0.1];
% \node at (-.8,.3) {$v_1$};
% \draw [fill] (-.3,0) circle [radius=0.1];
% \node at (-.3,.3) {$v_2$};
% \draw [fill] (.5,0) circle [radius=0.1];
% \node at (.5,.3) {$v_3$};
% \node at (1.1,.3) {$\cdots$};
% \draw [fill] (2.1,0) circle [radius=0.1];
% \node at (2.1,.3) {$v_n$};
% %
% \draw[<->,thick,red,dashed] (-.8,-.2)--(2.1,-.2);
% % \draw[<->,thick,red,dashed] (-.8,.6)--(2.1,.6);
% \draw[<->,thick,red,dashed] (-1,0)--(-1,3.8);
% \draw [fill, color=blue!50] (-.8,3.9) circle [radius=0.1];
% \node at (-.5,3.9) {$p$};
% \node[color=red!60] at (0.3,-.4) {$r=v_n-v_1$};
% % \node[color=red!60] at (1.5,.6) {$^{r}/_{2}$};
% \node[rotate=90,color=red!60] at (-1.2,1.9) {$2r$};
% %
% \draw[-,color=blue!40,dashed] (-1.1,0)--(-.8,3.9);
% \draw[-,color=blue!40,dashed] (.2,0)--(-.8,3.9);
% \draw[-,color=blue!40,dashed] (.8,0)--(-.8,3.9);
% \draw[-,color=blue!40,dashed] (2.4,0)--(-.8,3.9);
% \node[rotate=-83,color=blue!80] at (-.5,1.) {$w_1$};
% \node[rotate=-77,color=blue!80] at (0.2,.9) {$w_2$};
% \end{tikzpicture}
    \vspace{-5mm}
    \caption{Illustration of an unfair \enet, highlighted with green circles}
    \label{fig:enet1}
    % \vspace{-7mm}
\end{figure}
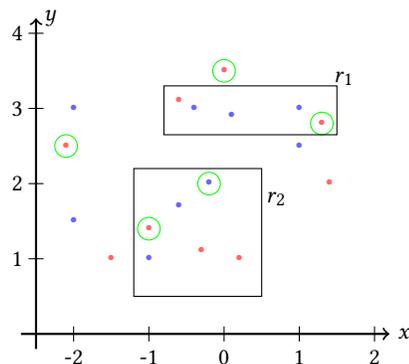

\begin{example}\label{ex:enet} {\sc (Part 1)}
    Consider a relation $T$, defined over two attributes $x$ and $y$, with the $n=18$ tuples shown as points in Figure~\ref{fig:enet1}.
    The range predicates in this relation are in the form of axis-parallel rectangles. 
    For example, $r_1$ in Figure~\ref{fig:enet1} corresponds with the following SQL query:

    % \vspace{-4mm}
    \begin{verbatim}
            SELECT * FROM T 
                WHERE -0.8<=x<=1.5 AND 2.6<=y<=3.3
    \end{verbatim}

    % \vspace{-5mm}
    In a very large setting, where the objective is to quickly identify a tuple matching the range query, 
    one can model the problem as 
    a range space $(\points,\ranges)$ where $\points$ is the set of points in the relation $T$ and $\ranges$ is the universe of all axis-parallel rectangles. 
    
    Let $\eps = \frac{5}{18}$.
    An \enet on $(\points,\ranges)$ is a subset of $T$ that guarantees to contain at least one point from any possible range with cardinality at least $5$.
    The points highlighted in green form such an \enet.
    For example, the rectangle $r_1$, which encompasses five points, has one point from the \enet.
    Using only the set of points in the \enet, one can quickly find a point satisfying any (sufficiently large) range query. For example, the highlighted point in the bottom-right of $r_1$ is the point in the \enet satisfying the above SQL query. \qed
\end{example}

\vspace{-1mm}
In \S~\ref{sec:application}, we will illustrate some of the other applications of \enets, including k-Nearest Neighbors and top-$k$ queries.

The traditional formulations of these problems, however, operate under the assumption of homogeneity in the data, neglecting the underlying (demographic) group structure that may be critical in fairness-sensitive applications. To better motivate this, let us consider Example~\ref{ex:enet} once more:

\noindent{\sc Example~\ref{ex:enet} (Part 2)}
{\it 
Suppose the tuples in Figure~\ref{fig:enet1} belong to two demographic groups specified by the color of the points (i.e., \at{\{blue,red\}}). One can notice the selected points for \enet mainly belong to the \at{red} group. As a result, answering range queries using this set will favor the red group by mostly returning a tuple from this group. \qed
}

% \color{blue}

As a concrete use case of example~\ref{ex:enet}, we can consider the {\em query answering on maps for finding points of interest}~\cite{netek2019performance, agafonkin2016supercluster, leafletjs}.
Specifically, let the x and y axes in the figure be the \at{longitude} and \at{latitude} of businesses such as restaurants on a map (Figure~\ref{fig:map-example} in the Appendix).
Now consider an interactive application, such as \href{https://yelp.com}{Yelp} or \href{https://maps.google.com}{Google Maps}, where users explore the map by panning and zooming to certain (rectangular) regions on their screen.

In order to make the UI interactive, the application needs to quickly find some points that match the specified region~\cite{agafonkin2016supercluster, chupurnoff2024displaying}. 
Using an \enet as the dataset representative, the system can quickly find and display some points that match the query, without scanning the entire dataset.
Now, if some groups are underrepresented in the \enet (e.g., only a few black-owned restaurants are selected), those groups will have a lower exposure (e.g., the black-owned restaurants are highlighted disproportionately less on the maps), causing potentially less income for those groups.

A \emph{fair} \enet ensures that the selected points are representative of all groups, and using it for approximate query answering is not discriminatory against some groups.

% \color{black}

\definecolor{rowgray}{gray}{0.95}
\renewcommand{\arraystretch}{1.2}
\begin{table*}[t]
\centering
\small
\begin{tabular}{|p{5cm}|l|p{5.5cm}|}
\hline
\rowcolor{gray!20}
\textbf{Problem} & \textbf{Fairness Notion} & \textbf{Proposed Algorithms} \\
\hline

\multirow{4}{*}{\textbf{Fair} \boldmath{$\varepsilon$}{\bf-net}} 
  & \multirow{3}{*}{Demographic Parity} 
    & Fair Monte-Carlo (\S~\ref{sec:epsnets:sample}) \\
  & & Fair Discrepancy-based (\S~\ref{sec:epsnets:disc})\\
  & & Fair Sketch-and-merge (\S~\ref{sec:epsnets:sketch})\\
  \cline{2-3}
  & Custom-ratios & Reduction to FGHS (\S~\ref{sec:epsnets:cr}) \\
\hline

\multirow{3}{*}{\textbf{Fair} \boldmath{$\varepsilon$}{\bf-sample}} 
  & \multirow{2}{*}{Demographic Parity} 
    & Fair Monte-Carlo (\S~\ref{sec:epssample}) \\
  & & Fair Sketch-and-merge (\S~\ref{sec:epssample}) \\
  \cline{2-3}
  & Custom-ratios & Infeasible (\S~\ref{sec:epssample}) \\
\hline

{\bf Fair Geometric Hitting Set (FGHS)} & {Custom-ratios} & {Reduction to DP Fair $\eps$-net (\S~\ref{sec:fghs})} \\
\hline
\end{tabular}
\vspace{1mm}
\caption{Summary of the proposed fairness-aware algorithms under different fairness measures.}
\label{tab:summary}
\vspace{-9mm}
\end{table*}

\vspace{-2.5mm}
\subsection{Technical Contributions}

In this paper, we introduce fairness on three foundational geometric approximation problems: the construction of \enets, $\eps$-samples, and the geometric hitting set problem. Our goal is to enforce fairness constraints while preserving strong approximation guarantees.

We formalize our problems based on two group-fairness notions: (a) demographic parity (DP), which requires preserving the group ratios from the input data distribution, and (b) custom-ratios fairness (CR), which generalizes DP by allowing arbitrary target ratios for different groups.

% We propose a range of algorithms for the studied problems:
% \abol{discuss our contributions...}
% \mohsen{I added the following (contribution):}

To address demographic parity in the construction of \enets, we propose two algorithms: a Monte Carlo Randomized algorithm based on sampling and a deterministic discrepancy-based algorithm. We further generalize both approaches to the weighted setting, where each point carries a weight and fairness is measured based on the total weight of each color group. %\abol{weight is not clear here}\mohsen{I added some detail.}. 

To address the custom-ratio fairness\footnote{It is important to note that any algorithm satisfying the more general custom-ratio (CR) fairness also satisfies demographic parity (DP); however, it may not do so efficiently.}, we reduce the fair \enet problem to the Fair Geometric Hitting Set problem. To solve this, we develop a fair LP-based algorithm that constructs geometric hitting sets (equivalently, geometric set covers),
% \mohsen{I changed: "we develop an LP-based algorithm" to "we develop a fair version of the LP-based algorithm". Because the LP-based algorithm is not ours. Feel free to change it back.}, 
satisfying both demographic parity and custom-ratio constraints. A summary of the proposed algorithms is provided in Table~\ref{tab:summary}.

In addition to the theoretical analysis, we conduct comprehensive experimental evaluations on real-world and synthetic datasets on several applications of \enets and hitting sets. Our experiments verified our theoretical findings since our algorithms' outputs satisfied the fairness requirements while minimally increasing the output size compared to the regular (unfair) outputs. Our experiments further demonstrate the efficiency of our algorithms across diverse settings, including different dataset sizes, dimensionalities, numbers of demographic groups, and their distribution patterns. %\abol{please check and make more accurate}\mohsen{I elaborated more on the last sentence.}

\vspace{-3mm}
\subsection{Paper Organization}
% \mohsen{To be re-written:}
% \begin{itemize}
%     \item \S~\ref{sec:prelim}: Gives a formal definition of the problem and preliminary concepts like fairness measures.
%     \item \S~\ref{sec:background}: We discuss the related background on the topic.
%     \item \S~\ref{sec:epsnets}: We discuss the fair \enet problem. Demographic parity is addressed by two algorithms: sampling-based and discrepancy-based.
%     \item \S~\ref{sec:epssample}: We discuss the fairness for $\eps$-samples. The DP is addressed similarly to \enets, and the feasibility of CFC is discussed with counter-examples.
%     \item \S~\ref{sec:setcover}: We discuss the fair version of the Geometric Set Cover.
% \end{itemize}

The remainder of the paper is structured as follows. Section~\ref{sec:prelim} introduces the necessary preliminaries and formal problem definitions. Section~\ref{sec:background} provides an overview of classical results on $\varepsilon$-nets, $\varepsilon$-samples, and the Geometric Hitting Set problem. In Section~\ref{sec:epsnets}, we present our algorithms for constructing fair $\varepsilon$-nets, followed by a discussion on fair $\varepsilon$-samples in Subsection~\ref{sec:epssample}. Section~\ref{sec:fghs} details our approach to the Fair Geometric Hitting Set problem. Related work is discussed in Section~\ref{sec:related}. Finally, Section~\ref{sec:exp} presents the experimental evaluation of the proposed algorithms.
\vspace{-2.5mm}
\section{Preliminaries}\label{sec:prelim}
In this section, we begin by introducing the necessary notation, followed by formal definitions of the preliminary concepts. We then provide a formulation of the problem studied in this paper.

% \renewcommand{\arraystretch}{1.2} % Add just before the tabular
% \begin{table*}[t]
% \centering
% \begin{tabular}{|l|l|l|}
% \hline
% \textbf{Problem} & \textbf{Fairness Measure} & \textbf{Proposed Algorithms} \\
% \hline

% \multirow{4}{*}{Fair $\eps$-net} 
%   & \multirow{3}{*}{Demographic Parity} & - Sampling-based \\
%   &                                     & - Discrepancy-based \\
%   &                                     & - Sketch-and-merge \\ \cline{2-3}
%   & Custom-ratios                       & - Reduction to FGSC \\ \cline{3-3}
% \hline

% \multirow{3}{*}{Fair $\eps$-sample} 
%   & \multirow{2}{*}{Demographic Parity} & - Sampling-based \\
%   &                                     & - Sketch-and-merge \\ \cline{2-3}
%   & Custom-ratios                       & - Infeasible (counter-example) \\ \cline{3-3}
% \hline

% Fair Geometric Set Cover (FGSC)
%   & Custom-ratios                       & - Reduction to Fair $\eps$-net (DP) \\
% \hline
% \end{tabular}
% \vspace{2mm}
% \caption{A summary of the proposed algorithms. \mohsen{Feel free to remove this table if you think it's not required. [ADD RUNTIMES]}}
% \label{tab:summary}
% \end{table*}

\vspace{-2.5mm}
\subsection{Notations}
\paragraph{Range spaces:} A range space $\rangespace$ consists of a finite set of points $\points$, where $|\points| = n$, and each point $\point_i$ ($1 \leq i \leq n$) lies in a $d$-dimensional space $\Re^d$. Associated with $\points$ is a family of subsets $\ranges$, referred to as ranges, where each $\range_j \in \ranges$ is a subset of $\points$. While the set of ranges $\ranges$ may be infinite in general, we denote by $m = |\ranges|$ the number of ranges when $\ranges$ is finite.

The {\it dual} of a range space $\rangespace$ is another range space, denoted by $(\ranges, \mathcal{X})$, in which each range in $\ranges$ is treated as a point, and the new family of ranges $\mathcal{X}$ is defined as below:
\[
    \mathcal{X}= \{\ranges(\point)\mid \point \in \points\},
\]
where $\ranges(\point) = \{ \range \in \ranges \mid \point \in \range \}$. That is, each point $\point$ in the original space defines a range in the dual space, consisting of all ranges in $\ranges$ that contain $\point$.

\paragraph{Demographic groups: }
We assume a fixed set of demographic groups, represented as a finite set of colors $\groups = \{\group_1, \group_2, \ldots, \group_k\}$.
Each point $\point_i \in \points$ is associated with a color $\group(\point_i) \in \groups$. This induces a partition of the point set $\points$ into demographic groups: for each $c \in \groups$, let $\points_c = \{ \point \in \points \mid \group(\point) = c \}$ denote the subset of points of color $c$.
Generally, for every subset $Y\subseteq X$, we set $Y_{\group}=\{p\in Y\mid \group(p)=\group\}$ for each $\group\in\groups$. We use the terms `demographic group' and `color' interchangeably.
% \begin{itemize}
%     \item 
%     {\bf Point coloring:} Each point $\point_i \in \points$ is associated with a color $\group(\point_i) \in \groups$. This induces a partition of the point set $\points$ into demographic groups: for each $c \in \groups$, let $\points_c = \{ \point \in \points \mid \group(\point) = c \}$ denote the subset of points of color $c$.
%     \item
%     {\bf Range coloring:} Alternatively, a color (demographic group) may be assigned to each range $\range_j \in \ranges$, denoted by $\group(\range_j) \in \groups$. Let $\ranges_c = \{ \range \in \ranges \mid \group(\range) = c \}$ be the set of ranges of color $c$.
% \end{itemize}

\paragraph{Weights:}
In the weighted setting, each point $\point_i \in \points$ is associated with a weight denoted by $w_i$. For any subset $S \subseteq \points$, the total weight of the points in $S$ is given by $w(S) = \sum_{\point_i \in S} w_i$.

\subsection{Definitions}\label{sec:prelim:def}
For the range space $\rangespace$, a subset $S \subseteq \points$ is said to be {\it shattered} by $\ranges$ if, for every subset $T \subseteq S$, there exists a range $\range \in \ranges$ such that $\range \cap S = T$. Define $\ranges_{|T}$ as the set of ranges induced on $T$, in other words, $\ranges_{|T} = \{\range \cap T \mid \range \in \ranges\}$.

The {\it VC dimension} of the range space $\rangespace$, denoted by $\VC{\rangespace}$, is the maximum size of a subset $S \subseteq \points$ that can be shattered by $\ranges$. If arbitrarily large shattered subsets exist, the VC dimension is said to be infinite. We will use $d$ as the VC dimension where the range space is clear from the context.

In the following, we formally define $\eps$-net, $\eps$-sample, and the Geometric Hitting Set problem~\cite{haussler1986epsilon, bronnimann1994almost}.

\begin{definition}[$\eps$-net]\label{def:epsnet}
    Fix a value $\eps \leq 1$. A subset $\epsnet \subseteq \points$ is called an $\eps$-net of $\rangespace$, if for any range $\range \in \ranges$ we have:
    \[
        \frac{|\range|}{|\points|} \geq \eps \Longrightarrow |\epsnet \cap \range| \geq 0
    \]
\end{definition}

\begin{definition}[$\eps$-sample]\label{def:epssample}
    Fix a value $\eps \leq 1$. A subset $\epssample \subseteq \points$ is called an $\eps$-sample of $\rangespace$, if for any range $\range \in \ranges$ we have:

    \[
        \left|\frac{|\epssample \cap \range|}{|\epssample|} - \frac{|\points \cap \range|}{|\points|}\right| \leq \eps
    \]
\end{definition}

% \abol{we have not defined the dual space but we use it in the definitions. We can either explain it in the preliminary, appendix, or refer the reader to prior work. But we have to provide some explanation for all terms}\mohsen{We have a short definition in the second paragraph of Notations. We can bring it into definitions section.}
\begin{definition}[Geometric Hitting Set]
    Given a range space $\rangespace$ with a bounded $VC$-dimension. The {\it Geometric Hitting Set} problem asks to find the smallest subset of points $\points^{*} \subseteq \points$ to hit all the ranges in $\ranges$, i.e.:

    \[
        \forall \range \in \ranges,\; \exists \point \in \points^* \text{ such that } \point \in \range
    \]
    $\points^*$ is called the smallest {\it hitting set} of $\rangespace$.
\end{definition}

The dual of this problem, formulated over the dual range space, is referred to as the {\sc Geometric Set Cover} problem (where the dual VC-dimension is bounded). Consequently, most of the results presented in this paper for Geometric Hitting Set also extend to the Set Covering setting.

\vspace{-2.5mm}
\subsection{Problems}
We now introduce fair variants of the problems by incorporating fairness constraints into their definitions.

\begin{definition}[Fair $\eps$-net]
    Given a range space $\rangespace$, a finite set of colors $\groups=\{\group_1,\ldots, \group_k\}$ representing $k$ groups, a parameter $\eps\in(0,1)$, and a vector $\mathcal{T} = (\tau_1, \tau_2, \cdots,\tau_k)$ of $k$ ratios for the groups, such that for every $l\in[k]$, $\tau_l \leq 1$ and $\sum_{l\in[k]} \tau_l = 1$, the goal is to compute a set $\epsnet\subseteq X$ of minimum size such that $\epsnet$ is an $\eps$-net of $\rangespace$ and for every $l\in[k]$,
    %Let $\epsnet$ be an $\eps$-net of $\rangespace$. We call this a {\bf fair} $\eps$-net with respect to $t$ if it maintains the ratio $\tau_l$ on color $\group_l$:

    \vspace{-7mm}
    \[ \hspace{20mm}
        \frac{|\epsnet \cap \points_{\group_l}|}{|\epsnet|} = \tau_l.
    \]    
\end{definition}
The Fair $\eps$-sample problem can be defined similarly.

\begin{definition}[Fair $\eps$-sample] 
    Given the same input as in the Fair $\eps$-net problem,
    the goal is to compute a set $\epssample\subseteq X$ of minimum size such that $\epssample$ is an $\eps$-sample of $\rangespace$ and for every $l\in[k]$,
    \[
        \frac{|\epssample \cap \points_{\group_l}|}{|\epssample|} = \tau_l.
    \]
\end{definition}

\begin{definition}[Fair Geometric Hitting Set, FGHS]
    Given a range space $\rangespace$ with a bounded $VC$-dimension and the color ratios $\mathcal{T} = (\tau_1, \tau_2, \cdots, \tau_k)$ defined similarly to the Fair $\eps$-net problem. The {\it Fair Geometric Hitting Set (FGHS)} problem\footnote{Note that the Geometric Set Cover and Geometric Hitting Set problems are equivalent by duality. The solutions we propose for Fair Hitting Set naturally extend to {\it Fair Geometric Set Cover (FGSC)}, where the groups are defined over the sets.} asks for the {\it smallest} subset $\points^* \subseteq \points$ satisfying the following conditions:
    \begin{enumerate}
        \item $\points^*$ is {\bf fair} with respect to $\mathcal{T}$, i.e.,
        \(
            \frac{|\points^* \cap \points_{\group_l}|}{|\points^*|} = \tau_l, \forall l\in[k].
        \)
        \item $\points^*$ hits all the ranges in $\ranges$.
    \end{enumerate}    
\end{definition}

\vspace{-3mm}
\subsection{Fairness}
In all the above problems, we may either address the {\it demographic parity} or {\it custom-ratio} fairness constraints defined as below.

We say that a subset of points $S \subseteq \points$ satisfies {\bf demographic parity (DP)} if, for each color, it maintains the same ratio as the ground set $\points$. In other words, the ratios $\mathcal{T}$ are defined as

\[
    \tau_l = \frac{|\points_{\group_l}|}{|\points|},\quad \forall l\in[k].
\]

The same definition is valid for a subset of ranges (instead of points) if the demographic groups are defined on ranges.

In a more general setting, we consider any arbitrary ratios $\mathcal{T}$ and try to satisfy the fairness accordingly; we call it {\bf custom-ratio (CR)} constraint. However, it is sometimes impossible to satisfy {\it any ratio} in some cases.
% \stavros{We should add an example here showing why it is impossible or some assumptions we should satisfy as we had in the KDD paper.}

\begin{example}
    Consider a range space consisting of 10 points, where 9 are blue and only 1 is red. Suppose that any valid hitting set must contain more than two points. In this case, it is impossible to satisfy a 50\%-50\% color ratio in any hitting set, as there are insufficient red points to meet the fairness constraint.
\end{example}
% \subsection{Overview}

% Table~\ref{tab:summary} provides a summary of the proposed algorithms. For each problem, we develop methods that satisfy either Demographic Parity or Custom-Ratio fairness constraints or show the infeasibility with counter-examples. In certain cases, problems can be reduced to one another.
\vspace{-2.5mm}
\section{Background}\label{sec:background}
In this section, we briefly review classical results related to the problems defined in Section~\ref{sec:prelim:def}. References to this background will be made in subsequent sections as needed.

\paragraph{$\eps$-nets and $\eps$-samples:}
A classical approach to constructing $\eps$-nets is based on random sampling:

\begin{theorem}[$\eps$-net~\cite{haussler1986epsilon}]\label{thm:epsnet}
    Given a range space $\rangespace$ with bounded VC dimension $d$. A random sample (with replacement) of size:

    \[
        \lambda  \geq \max\left(\frac{4}{\eps}\log\frac{4}{\prob}, \frac{8d}{\eps}\log\frac{16}{\eps}\right)
    \]

    gives us an $\eps$-net with probability at least $1 - \prob$.
\end{theorem}

A similar result characterizes the number of random samples required to construct an $\eps$-sample:

\begin{theorem}[$\eps$-sample~\cite{VC71}]\label{thm:epssample}
    Given a range space $\rangespace$ with bounded VC dimension $d$. A random sample (with replacement) of size:
    \[
        \gamma \geq \frac{c_0}{\eps^2}\left(d\log\frac{d}{\eps} + \log\frac{1}{\prob}\right)
    \]

    gives us an $\eps$-sample with probability at least $1 - \prob$, for a large enough constant $c_0$.
\end{theorem}

Deterministic constructions of $\eps$-nets and $\eps$-samples can also be achieved via discrepancy methods~\cite{chazelle2000discrepancy, chazelle1996linear}, which we discuss in the relevant section, in the context of building fair $\eps$-nets deterministically.

\paragraph{Geometric Hitting Set:} The Hitting Set and Set Cover problems are dual and algorithmically equivalent, which is known to be NP-hard. In the general (non-geometric) setting, a standard greedy algorithm—which iteratively selects the point that hits the most (remaining) ranges—yields a $\log n$ approximation. However, in geometric settings with bounded VC dimension $d$, improved algorithms achieve a $\log d$ approximation factor.

To achieve this improvement, one approach uses the multiplicative weight update method~\cite{arora2012multiplicative, har2011geometric}, while another uses LP relaxation, reducing the Hitting Set problem to an instance of the \emph{weighted} $\varepsilon$-net problem~\cite{lon01, ch09}.

\paragraph{\textbf{Assumptions}}
Similarly to~\cite{chazelle2000discrepancy}, we assume that $n=|X|=2^\xi$ for a positive integer $\xi$, i.e., the number of points in $X$ is a power of $2$. Furthermore, we make the mild assumption that for every $\group\in\groups$, $|X_\group|=2^{\xi_\group}$, where $\xi_\group$ is a positive integer.
Our goal in all cases is to construct fair $\eps$-nets of small size -- ideally comparable to the size of unfair $\eps$-nets. In order to achieve this goal, 
%we assume that each color is sufficiently represented in the original point set $X$ (DP constraints) or . More specifically, in all cases
we assume that $\tau_\ell\cdot \sizenet\geq 1$, for every $\ell\in[k]$.
If these inequalities do not hold, then the size of a fair $\eps$-net might be much larger (even $k$ times larger) than the size of the (unfair) $\eps$-net.
Even if this assumption does not hold, then all our algorithms are correct, however the approximation factor and the running time of some of our algorithms might increase by a factor $\frac{1}{\sizenet\cdot \min_{\ell\in[k]}\tau_\ell}$.
\begin{table*}[t]
\centering
\small
\begin{tabular}{|l|p{5cm}|p{6cm}|}
\hline
\rowcolor{gray!20}
\textbf{Aspect} & \textbf{Fair Monte-Carlo} & \textbf{Fair Discrepancy-based} \\
\hline
Size Increase Factor & \( O(\log k) \) & Constant $O(1)$ \\
\hline
\multirow{2}{*}{Running time} & \multirow{2}{5cm}{$O(n)$} & $O(m\cdot n\log n)$ for simple discrepancy \\
\cline{3-3}
 & & $O\left(n\cdot\frac{d^{3d}}{\eps^{2d}}\log^d\left(\frac{d}{\eps}\right)\right)$ for Sketch-and-Merge \\
\hline
Weighted Case & Similar guarantees to unweighted & Very slow (depends on $\sum_i  w_i$) \\
\hline
Determinism & Monte Carlo Randomized (always fair if successful) & Deterministic if $\ranges$ is materialized (always fair) \\
\hline
Implementation & Easier to implement & More complex \\
\hline
\end{tabular}
\vspace{1mm}
\caption{High-level comparison of the fair Monte-Carlo sampling-based and the fair discrepancy-based algorithms for addressing DP in $\eps$-nets.}
\label{tab:eps-net-dp}
\vspace{-8mm}
\end{table*}

\section{Fair $\eps$-nets}\label{sec:epsnets}
In this section, we discuss our algorithms for the Fair $\eps$-net and $\eps$-sample problems. 
First, in Sections~\ref{sec:demop} and \ref{sec:demop2}, we address fair $\eps$-net under the demographic parity, followed by the custom-ratios fairness in \S~\ref{sec:custom}. We will then study fair $\eps$-sample in \S~\ref{sec:epssample}.

We propose two categories of algorithms for constructing fair $\eps$-nets that satisfy demographic parity: a {\it Monte-Carlo sampling-based} and a {\it discrepancy-based} algorithm. The sampling-based approach returns a larger fair $\eps$-net efficiently, while the discrepancy-based method algorithms return a smaller fair $\eps$-net at the expense of increased runtime.
We design two discrepancy-based algorithms. While they both return a fair $\eps$-net of (asymptotically) the same size, the second one is more efficient when $d$ is small and $m$ is large.
A comparative summary of all our algorithms is presented in Table~\ref{tab:eps-net-dp}.

%   \abol{do we not want to mention the obvious baseline approach that provides a $O(k)$ approximation?}

\vspace{-2.5mm}
\subsection{A Monte-Carlo Randomized Algorithm for Satisfying Demographic Parity}\label{sec:demop}\label{sec:epsnets:sample}
% \subsubsection{Sampling-based Algorithm}

Our first algorithm is a sampling-based Monte Carlo randomized algorithm that finds a fair $\eps$-net satisfying demographic parity with a high probability.
At a high level, the algorithm is developed based on the observation that a random sample from $X$ should be ``near-fair'' for each color $c_i$, in the sense that the ratio of samples from $c_i$ in the sample set should be close to the required ratio $\tau_i$. Hence, with a high probability, an $O(\log k)$-factor increase in the size of the set is enough to satisfy the demographic parity ratios.

%\stavros{I will re-write subsection 4.1.1}

\vspace{-2mm}
\paragraph{Algorithm}
% Let $\sizenet\geq \max\{\frac{4}{\eps}\log\frac{4}{\prob/2}, \frac{8d}{\eps}\log\frac{16}{\eps}\}$.
Draw a set $\epsnet$ of $\sizenet\geq \max\{\frac{4}{\eps}\log\frac{4}{\prob/2}, \frac{8d}{\eps}\log\frac{16}{\eps}\}$ uniform random samples (with replacement) from $X$.
Let $v=2\ln(k\cdot\frac{1}{\prob/2})$.
% For every color $\group\in \groups$ let $\epsnet^A_\group$ be a set of $\max\{v\cdot \frac{|X_\group|}{|X|}\cdot \sizenet-|X_\group\cap\epsnet|,0\}$ arbitrary points of color $\group$. 
% Let $\res$ be the union of $\epsnet$ with the additional arbitrary points for each color. That is,
% \[
% \res = \epsnet 
% \]
For every color $\group\in \groups$ add $\max\{(1+v)\cdot \frac{|X_\group|}{|X|}\cdot \sizenet-|\epsnet_\group|,0\}$ arbitrary points of color $\group$ to $\epsnet$. Let $\res$ be the new set after we traversed all colors $\group\in \groups$. Return $\res$. A pseudo-code of this algorithm is presented in Algorithm~\ref{alg:epsnet:dp:sampling}.

\vspace{-2mm}
\paragraph{Correctness and runtime analysis}

\begin{lemma}
\label{lem:size}
    For every color $\group\in\groups$, $|\epsnet_\group|\leq (1+v)\cdot\frac{|X_\group|}{|X|}\cdot \sizenet$, with probability at least $1-\prob/2$. 
\end{lemma}

\vspace{-4mm}
\begin{proof}
We fix a color $\group\in \groups$.
Let $\mathbb{I}_{\group, i}$ be an indicator random variable which is equal to $1$ if and only if the $i$-th sample is a point with color $\group$. We note that $|\epsnet_\group|=\sum_{i\in[\sizenet]}\mathbb{I}_{\group, i}$.
We also have $\mathbb{E}[\mathbb{I}_{\group, i}]=\frac{|X_\group|}{|X|}$ and we define $\mu_\group=\mathbb{E}[\sum_{i\in[m]}\mathbb{I}_{\group, i}]=\sizenet\cdot \frac{|X_\group|}{|X|}$.
%Let $\delta=v$.
Using the Chernoff's inequality,\footnote{Suppose $Y_1,\ldots, Y_\gamma$ are independent random variables taking values in $\{0,1\}$. Let $Y$ denote their sum and let $\mu=\mathbb{E}[Y]$. Then for any $\delta>0$, the multiplicative Chernoff inequality states that $\mathbb{P}[Y\geq (1+\delta)\mu]\leq \mathsf{exp}(-\delta^2\mu/(2+\delta))$.}
\begin{align*}
    &\mathbb{P}\Big(|\epsnet_\group| \geq (1 + v)\mu_\group\Big) \leq \mathsf{exp}\left(\frac{-v^2 \mu_{\group}}{2 + v}\right).
\end{align*}
By the assumptions on the ratios, $\mu_\group\geq 1$, so
\begin{align*}
    \mathbb{P}\Big(|\epsnet_\group| \geq (1 + v)\frac{|X_\group|}{|X|}\cdot \sizenet\Big)& \leq \mathsf{exp}\left(\frac{-v^2}{2 + v}\right)
    =\mathsf{exp}\left(\frac{-4\ln^2(\frac{k}{\prob/2})}{2+2\ln(\frac{k}{\prob/2})}\right)\\
    &\leq \mathsf{exp}\left(-\ln(\frac{k}{\prob/2})\right)
    =\frac{\prob}{2k}.
\end{align*}

\vspace{-2mm}
By applying the union bound,
\begin{align*}
\mathbb{P}\Big(\exists \group \in \groups,\; |\epsnet_{\group}| \geq (1 + v)\frac{|X_\group|}{|X|}\cdot \sizenet\Big)&\leq \frac{\prob}{2}.
\end{align*}
Hence, we conclude that for every color $\group\in \groups$, $|\epsnet_\group|\leq (1+v)\frac{|X_\group|}{|X|}\cdot \sizenet$ with probability at least $1-\prob$.
\end{proof}

\begin{lemma}
    The set $\res$ is a fair $\eps$-net of $\rangespace$ and $|\res|=O((1+v)\sizenet)$, with probability at least $1-\prob$.
\end{lemma}
\begin{proof}
 From Theorem~\ref{thm:epsnet} (in Background section), we have that $\epsnet$ is an $\eps$-net  of the range space $\rangespace$ with probability at least $1-\prob/2$. We have that $\res\supseteq \epsnet$, so $\res$ is an $\eps$-net with probability at least $1-\prob/2$.

 From Lemma~\ref{lem:size}, we have that for every group $\group\in \groups$, $|\epsnet_\group|\leq (1+v)\frac{|X_\group|}{|X|}\cdot \sizenet$ with probability at least $1-\prob/2$. The algorithm adds $(1+v)\frac{|X_\group|}{|X|}\sizenet-|\epsnet_\group|$ elements of of color $\group$ to create the set $\res$. Hence,
 \begin{align*}
        |\res| &= |\epsnet| + \sum_{\group \in \groups}\left( (1+v) \cdot \frac{|X_\group|}{|X|}\sizenet - |\epsnet_\group|\right)\\
        &= |\epsnet| + (1+v)\sum_{\group \in \groups} \frac{|X_\group|}{|X|}\sizenet - \sum_{\group\in\groups}|\epsnet_\group|\\
        &= \sizenet + (1+v)\sizenet - \sizenet=(1+v)\sizenet,
\end{align*}
with probability at least $1-\prob/2$.

Finally, we show that $\res$ is a fair $\eps$-net.
For every $\group\in\groups$ we have
$\frac{|\res_\group|}{|\res|}=\frac{(1+v)\sizenet\frac{|X_\group|}{|X|}}{(1+v)\sizenet}=\frac{|X_\group|}{|X|}$.
The result follows.
\end{proof}

\begin{lemma}
    The algorithm's time complexity is $O(|X|)$.
\end{lemma}
\begin{proof}
We place all points from $X$ in a table. We sample $\sizenet$ indexes and we add the corresponding points in $\epsnet$. Then we go through each element in the table and we add it in $\epsnet$ if needed in $O(1)$ time.
\end{proof}

\begin{algorithm}[t]
\begin{algorithmic}[1]
\small
\Require Range space $\rangespace$, Epsilon $\eps$.
\Ensure The DP fair $\eps$-net $\res$, found by \underline{sampling}.
\Function{FMC}{$\points, \ranges, \eps$}
    \State $v\gets 2\ln(k\cdot\frac{1}{\prob/2})$
    \State $\epsnet \gets $ a random subset of size $\lambda$ from $\points$
    % \While{$\epsnet$ not satisfying Lemma~\ref{lem:size}}\Comment{Until success}
        % \State $\epsnet \gets $ a new random subset of size $\Delta$ from $\points$
    % \EndWhile
    \State $\res \gets \textit{copy}(\epsnet)$\Comment{Initialize a new set for fair $\eps$-net.}
    \For{All colors $\group \in \groups$}
        \State $Tmp \gets$ arbitrary $(1+v) \frac{|\points_{\group}|}{|\points|}\sizenet - | \epsnet_\group|$ points from $\points_{\group}$
        \State $\res \gets \res \cup Tmp$
    \EndFor
    \State \Return $\res$\Comment{This is a fair $\eps$-net with probability $\geq 1-\prob$.}
\EndFunction
\caption{Fair Monte-Carlo (FMC) algorithm for building an $\eps$-net satisfying DP fairness.}
\label{alg:epsnet:dp:sampling}
\end{algorithmic}
\end{algorithm}

\vspace{-2mm}
\begin{theorem}
\label{thm:sampling}
    Given a range space $\rangespace$, with $|X|=n$, VC dimension $d$, DP constraints $\mathcal{T}$, and parameters $\eps,\prob\in(0,1)$,
    there exists a randomized algorithm that constructs a fair $\eps$-net with respect to $\mathcal{T}$ of size $O(\frac{1}{\eps}\max\{\log\frac{1}{\prob},d\log\frac{1}{\eps}\}\cdot\log(\frac{k}{\prob}))$ with probability at least $1-\prob$, in $O(n)$ time, where $k$ is the number of different colors in set $X$.
\end{theorem}

\subsubsection{Weighted fair $\eps$-net}
\label{subsec:weightedfairepsnet}
Our randomized algorithm can be extended to the fair $\eps$-net problem over a weighted set of points.
Assume that each point $p_i\in X$ is associated with a weight $w_i$. Let $W_\group=\sum_{p_i\in X_{\group}}w_i$ be the sum of weights of points with color $\group\in\groups$. For simplicity and without loss of generality, we assume that $\sum_{\group\in\groups}W_\group=1$.
A weighted $\eps$-net is a subset $\res\subseteq X$ such that for every $R\in\ranges$ with $\sum_{p_i\in R}w_i\geq \eps$, it holds that $R\cap \res\neq \emptyset$.
In terms of fairness,
the DP ratios $\mathcal{T}$ are defined as $\tau_\ell=W_{\group_\ell}$, for each $\ell\in[k]$.
The goal is to compute a set $\res\subseteq X$ such that, for every $R\in\ranges$ with $\sum_{p_i\in R}w_i\geq \eps$, it holds that $R\cap \res\neq \emptyset$ (weighted $\eps$-net) and for every $\group_\ell\in\groups$, $\frac{|\res_{\group_\ell}|}{|\res|}=\tau_\ell$.
As we had in the unweighted case, we assume for simplicity that the ratios in $\mathcal{T}$ are chosen in that way such that $W_{\group_\ell}\cdot\sizenet=\tau_\ell\cdot\sizenet\geq 1$.

%Let $\mu_\mathcal{T}=\min_{\ell\in[k]}W_\ell=\tau_\ell$ be the minimum sum of weights of the points with the same color.
It is known~\cite{haussler1986epsilon}, that a set of $\sizenet$ random samples from $X$ with respect to the weights of the points returns a weighted (unfair) $\eps$-net. Hence, intuitively, we can extend the algorithm from the unweighted case to the weighted case, as follows.

\paragraph{Algorithm} Draw a set $\epsnet$ of $\sizenet$ random samples with respect to the weights of the points from $X$.
Let $v=2\ln(k\cdot\frac{1}{\prob/2})$.
%Let $v=\frac{2}{\sizenet\cdot \mu_\mathcal{T}}\ln(k\cdot\frac{1}{\prob/2})$.
For every color $\group\in \groups$, add $\max\{(1+v)\cdot \frac{|X_\group|}{|X|}\cdot \sizenet-|\epsnet_\group|,0\}$ arbitrary points of color $\group$ to $\epsnet$. Let $\res$ be the new set after we traversed all colors $\group\in \groups$. Return $\res$.
%The main difference with the algorithm in the unweighted case, is that $v$ is a larger value that depends on $\mu_\mathcal{T}$.

Using the same arguments as in Lemma~\ref{lem:size}, we can show that
$\mathbb{P}\Big(\exists \group \in \groups,\; |\epsnet_{\group}| \geq (1 + v)\frac{|X_\group|}{|X|}\cdot \sizenet\Big)\leq \frac{\prob}{2}$. Skipping the details, we conclude with the following theorem.
\begin{theorem}
\label{thm:weightedsampling}
    Given a range space $\rangespace$ over a weighted set of points $X$, with $|X|=n$, VC dimension $d$, DP constraints $\mathcal{T}$, and parameters $\eps,\prob\in(0,1)$,
    there exists a randomized algorithm that constructs a weighted fair $\eps$-net with respect to $\mathcal{T}$ of $O\!\left(\!\frac{1}{\eps}\max\{\log\frac{1}{\prob},d\log\frac{1}{\eps}\}\!\!\cdot\!\log\frac{k}{\prob}\!\right)$ size
     with probability at least $1-\prob$, where $k$ is the number of different colors in set $X$. The time complexity of the algorithm is $O(n)$.
\end{theorem}
%\begin{theorem}
%\label{thm:weightedsampling}
%    Given a range space $\rangespace$ over a weighted set of points $X$, with $|X|=n$, VC dimension $d$, DP constraints $\mathcal{T}$, and parameters $\eps,\prob\in(0,1)$, there exists a randomized algorithm that constructs a weighted fair $\eps$-net with respect to $\mathcal{T}$ of $O\left(\left(1+\frac{1}{\sizenet\cdot \mu_{\mathcal{T}}}\cdot\ln(\frac{k}{\prob})\right)\cdot \sizenet\right)$ size with probability at least $1-\prob$, where $\sizenet=\frac{1}{\eps}\max\{\log\frac{1}{\prob},d\log\frac{1}{\eps}\}$, $k$ is the number of different colors in set $X$, and $\mu_\mathcal{T}=\min_{\ell\in[k]}\tau_\ell$. The running time of the algorithm is $O(n)$.
%\end{theorem}

\vspace{-2.5mm}
\subsection{Discrepancy-based Algorithms for Satisfying Demographic Parity}\label{sec:demop2}\label{sec:epsnets:disc}
% \subsubsection{Discrepancy-based Algorithm}

The Monte-Carlo sampling-based algorithm proposed in the previous section, while being efficient, returns a fair $\eps$-net, which is larger than standard $\eps$-nets by roughly a $\log(k)$ factor. Besides, similar to other randomized Monte-Carlo algorithms, it does not always find a valid solution.
In this section, we propose two fair discrepancy-based algorithms to address these issues; we propose deterministic methods that construct smaller fair $\eps$-nets. Such results come at the cost of increased running time.
While both discrepancy-based algorithms return a fair $\eps$-net of (asymptotically) the same size, the first one is more efficient when $m$ is small and $d$ is large, while the second one is more efficient when $d$ is small and $m$ is large.
% The high level idea of our algorithm is the following.

\subsubsection{Discrepancy-based algorithm for fair $\eps$-net}
\label{subsec:discr1}
We first describe the notion of discrepancy to construct an $\eps$-net. After that, we adjust this method to construct a {\bf fair} $\eps$-net. More details on the definitions and the proofs relating to the standard $\eps$-net construction using discrepancy can be found in ~\cite{chazelle2000discrepancy, har2011geometric}.

Before we start the description of our algorithm, we give some useful definitions.
Let $\rangespace$ be a range space with $|X|=n$ and $|\ranges|=m$.
Define a coloring function $\kappa: \points \to \{-1, +1\}$ that assigns a color $+1$ or $-1$ to each point in $\points$.
The coloring $\kappa$ differs from the demographic groups and should not be mistaken with colors/groups $\groups = \{\group_1, \group_2, \cdots, \group_k\}$.
The {\it discrepancy} of $\kappa$ over a range $\range \in \ranges$ is defined as:
$|\kappa(\range)| = \left|\sum_{\point \in \range} \kappa(\point)\right|$.
The {\it discrepancy} of the function $\kappa$ is defined as:
$\mathsf{disc}(\kappa) = \max_{\range \in \ranges} |\kappa(\range)|$.
Namely, it is the color difference in the most unbalanced range based on coloring $\kappa(\cdot)$.
The discrepancy of the range space $\rangespace$ is defined as the best achievable discrepancy by a coloring function,
$\mathsf{disc}(\points) = \min_{\kappa: \points \to \{-1, +1\}} \mathsf{disc}(\kappa)$.

A {\it matching} $\Pi$ of $\points$ is a set of pairs $\{(p_i, p_j) \mid p_i, p_j \in \points\}$ that partitions the entire set $\points$ into $n / 2$ pairs. A coloring $\kappa$ is said to be {\it compatible} with a matching $\Pi$, if for any pair $(p_i, p_j) \in \Pi$, we have: $\kappa(p_i) + \kappa(p_j) = 0$.

The following procedure defines a major step in building $\eps$-nets using this method:

\paragraph{\textsf{Random Halving}}Let $\Pi$ be an \emph{arbitrary} matching on $\points$. For each pair $(p_i, p_j) \in \Pi$, randomly either color $p_i$ as $+1$ and $p_j$ as $-1$, or the other way around, by tossing a fair coin. Then, without loss of generality, only keep the black points. Let $\points^{(1)}$ be the result set with $|\points^{(1)}| = \frac{n}{2}$. 

%\begin{definition}[Random Halving]\label{def:halving} Let $\Pi$ be an \emph{arbitrary} matching on $\points$. For each pair $(p_i, p_j) \in \Pi$, randomly either color $p_i$ as black and $p_j$ as white, or the other way around, by tossing a fair coin. Then, just consider all the points colored as black (wlog white). Let $\points^{(1)}$ be the result set: $|\points^{(1)}| = \frac{n}{2}$. \end{definition}

The next lemma that uses the \textsf{Random Halving} procedure is proven in ~\cite{har2011geometric}.

\begin{lemma}
\label{lem:RH}
We are given a range space $\rangespace$ with $|X|=n$ and $|\ranges|=m$.
    Let $\Pi$ be an arbitrary matching of $X$ and let 
    $\kappa$ be the coloring constructed by applying the procedure \textsf{Random Halving}. With a probability more than $\frac{1}{2}$, we have:

    \vspace{-4mm}
    \[
        \forall \range \in \ranges, \quad |\kappa(\range)| \leq \Delta= \sqrt{2 n \ln(4m)}.
    \]
\end{lemma}

The construction of Lemma~\ref{lem:RH} can be made deterministic by the method of \emph{conditional expectations}~\cite{har2011geometric}, assuming that we have access to all the ranges in $\ranges$:
\newcommand{\hp}{\mathsf{Halving}}
\begin{lemma}
\label{lem:halving2}
    Given a range space $\rangespace$ with $|X|=n$ and $|\ranges|=m$.
    Let $\Pi$ be an arbitrary matching on $X$.
   There exists a deterministic method $\hp(X,\ranges,\Pi)$ that constructs a coloring function $\kappa:X\rightarrow \{-1,1\}$ which is compatible with matching $\Pi$ in $O(n\cdot m)$ time, such that

   \vspace{-7mm}
    \[
        \forall \range \in \ranges,\; |\kappa(\range)| \leq \Delta.
    \]
\end{lemma}

\paragraph{Algorithm}
Using the notion and properties of discrepancy, we describe our algorithm for constructing a fair $\eps$-net.
Our algorithm works in iterations $i=1,\ldots, U$, where $U=\log\frac{n}{c_0\frac{d}{\eps}\log\frac{d}{\eps}}$ and $c_0$ is a sufficiently large constant. Initially, $X^{(0)}=X$.
In the $i$-th iteration, we construct a \emph{fair matching} $\Pi_f^{(i-1)}$ on $X^{(i-1)}$, as follows. Let $Y=X^{(i-1)}$. For each pair $p_{j_1}, p_{j_2} \in Y$ with $\group(p_{j_1})=\group(p_{j_2})$, we add the pair $(p_{j_1},p_{j_2})$ in $\Pi_f^{(i-1)}$ and remove $p_{j_1},p_{j_2}$ from $Y$. We repeat the same pairing procedure until $Y=\emptyset$. Then, we run the halving procedure from Lemma~\ref{lem:halving2}, executing $\hp\left(X^{(i-1)}, \ranges_{|X^{(i-1)}}, \Pi_f^{(i-1)}\right)$. Let $X^{(i)}$ be the points with color $+1$ returned by the halving procedure.
In the end, we return $X^{(U)}$.

\paragraph{Correctness and runtime analysis}
We first show that the set returned by our algorithm is an $\eps$-net of small size. Then we show that this is also a fair $\eps$-net.
\begin{lemma}
\label{lem:epsnet1}
    Our algorithm returns an $\eps$-net of size $O(\frac{d}{\eps}\log\frac{d}{\eps})$.
\end{lemma}
\begin{proof}
Given the assumptions we discussed in the introduction (also similar assumptions were made in~\cite{dehghankar2024fair}), i.e., $|X_\group|$ is a power of $2$ for every group $\group\in \groups$, and there are enough points from every group, our algorithm finds fair matchings in every iteration and terminates.
It is known from~\cite{chazelle2000discrepancy} that if we apply Lemma~\ref{lem:halving2} $i^\star$ times (given an arbitrary pairing in every iteration), then $X^{(i^\star)}$ is an $\eps$-net of $X$, where $i^\star$ is a value such that

\vspace{-8mm}
\[\hspace{50mm}
        2^{i^\star} = \frac{n}{c_0 \frac{d}{\eps} \log \frac{d}{\eps}}.
    \]
    Notice that this iteration argument from~\cite{chazelle2000discrepancy} holds for arbitrary matchings. Hence, the fair matchings we construct in each iteration do not prevent us from using this argument.
    Our algorithm is executed for $U=\log\frac{n}{c_0\frac{d}{\eps}\log\frac{d}{\eps}}=i^\star$ iterations so $X^{(U)}$ is an $\eps$-net.
    By definition, notice that $|X^{(i)}|=\frac{|X^{(i-1)}|}{2}$, so $|X^{(U)}|=\frac{|X|}{2^U}=O(\frac{d}{\eps}\log\frac{d}{\eps})$. The lemma follows.
\end{proof}

\begin{lemma}
    Our algorithm returns a fair $\eps$-net.
\end{lemma}
\begin{proof}
    Form Lemma~\ref{lem:epsnet1}, we have that $X^{(U)}$ is an $\eps$-net. Hence we only need to show that $\frac{|X^{(U)}_\group|}{|X^{(U)}|}=\frac{|X_\group|}{|X|}$ for every $\group\in\groups$. In every iteration $i$ of our algorithm we construct a fair patching $\Pi_f^{(i-1)}$ over the set $X^{(i-1)}$. The set $X^{(i)}$ is constructed from $X^{(i-1)}$ keeping exactly one point from every pair in $\Pi_f^{(i-1)}$. Since both points in every pair in $\Pi_f^{(i-1)}$ has the same color, we have that $$\frac{|X^{(i)}_\group|}{|X^{(i)}|}=\frac{\frac{1}{2}|X^{(i-1)}_\group|}{\frac{1}{2}|X^{(i-1)}|}=\frac{|X^{(i-1)}_\group|}{|X^{(i-1)}|},$$ for every $\group\in\groups$. By repeating this argument for every $i=1,\ldots, U$, we conclude that $\frac{|X^{(U)}_\group|}{|X^{(U)}|}=\frac{|X_\group|}{|X|}$ for every $\group\in\groups$.
\end{proof}

\begin{lemma}
    The running time of our algorithm is $O(n\cdot m\cdot \log n)$.
\end{lemma}
\begin{proof}
In each iteration $i$ of our algorithm, we construct a fair matching in $O(|X^{(i-1)}|)=O(n)$ time by making a pass over all points in $X^{(i-1)}$.
    The Halving procedure from Lemma~\ref{lem:halving2} runs in $O(|X^{(i-1)}|\cdot |\ranges_{|X^{(i-1)|}}|)
    =O(n\cdot m)$ time. 
    We execute $U=O(\log n)$ iterations of the algorithm so the total running time is  $O(n\cdot m\cdot \log n)$.
\end{proof}

Putting everything together, we conclude with Theorem~\ref{thm:discrepancy}.

\begin{theorem}
\label{thm:discrepancy}
    Given a range space $\rangespace$, with $|X|=n$, $|\ranges|=m$, VC dimension $d$, DP constraints $\mathcal{T}$, and a parameter $\eps\in(0,1)$,
    there exists a deterministic algorithm that constructs a fair $\eps$-net with respect to $\mathcal{T}$, of size $O(\frac{d}{\eps}\log\frac{d}{\eps})$ in $O(n\cdot m \cdot\log n)$ time.
\end{theorem}

%\begin{theorem}
 %   The above discrepancy-based algorithm, gives a fair $\eps$-net $\epsnet$ of size $\Theta(\frac{d}{\eps}\log \frac{d}{\eps})$. The running time is linear in the case of randomized construction. In this case, by repeating the whole process a constant number of times, one can increase the probability of success.
%\end{theorem}
A comparison of the two algorithms, discrepancy-based and sampling-based, is provided in Table~\ref{tab:eps-net-dp}.

\paragraph{Weighted case: } If the points $\points$ have weights, a simple approach would be to replicate each point according to its weight $w_i$ (assuming all the weights are integers; otherwise, we multiply them by a large value). However, the running time of this method would depend on $\sum_i w_i$ (integer weights), which can be significantly larger than $n$. The same argument is valid for space usage.

\subsubsection{Discrepancy with Sketch-and-Merge}\label{sec:demop3}\label{sec:epsnets:sketch}

In this subsection, we discuss an alternative deterministic algorithm that is more efficient than the algorithm from Theorem~\ref{thm:discrepancy} if the VC dimension is small. At a high level, the algorithm works as follows. We first construct a fair $\eps$-sample of small size satisfying DP constraints using a hierarchical bottom-up approach (resulting in a tree structure visualized in Figure~\ref{fig:partitiontree}). Then we apply the algorithm from Theorem~\ref{thm:discrepancy} on the fair $\eps$-sample set to construct a fair $\eps$-net. The algorithm is faster because, this time, the expensive algorithm from Theorem~\ref{thm:discrepancy} is executed on a small set, i.e., a fair $\eps$-sample.
%Using the constructed fair $\eps$-sample set, we construct a fair $\eps$-net solving an instance of the fair set cover problem.
\begin{figure}
    \centering
    \includegraphics[width=0.7\linewidth]{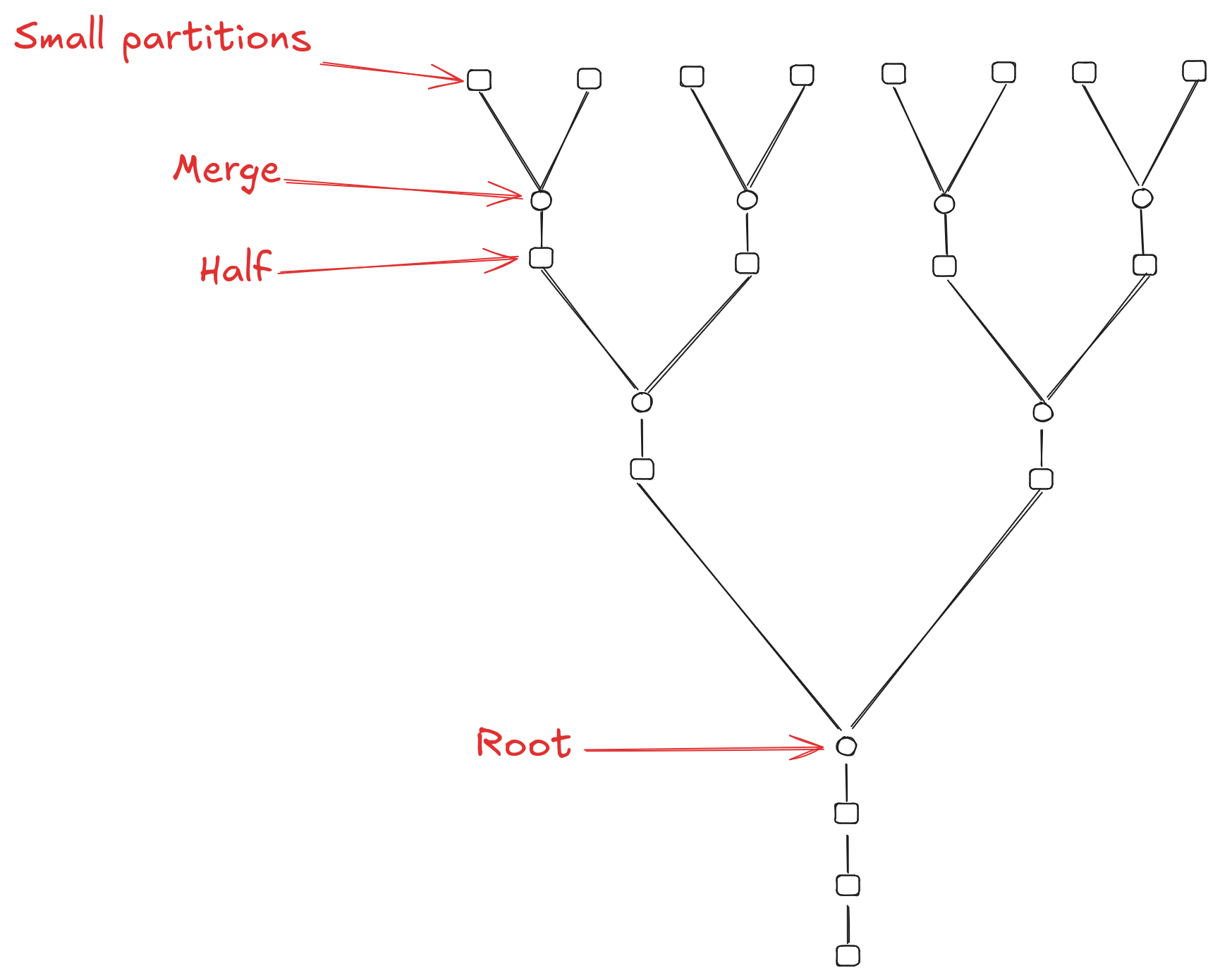}
    \vspace{-4mm}
    \caption{Sketch-and-Merge approach on a partition tree.}
    \label{fig:partitiontree}
    % \vspace{-5mm}
\end{figure}

\begin{algorithm}[h]
\begin{algorithmic}[1]
\small
\Require Range space $\rangespace$, Epsilon $\eps$, the value $2^p$ for the size of small partitions.
\Ensure The DP fair $\eps$-net $\res$, found by sketch-and-merge.
\Function{FSM}{$\points, \ranges, \eps, 2^p$}
    \State $\parti \gets $ partition $\points$ to $\frac{n}{2^p}$ equally-sized {\bf fair} subsets.
    \State $\mathcal{L} \gets \parti$ \Comment{All nodes in current level.}
    \While{$|\mathcal{L}| > 1$}\Comment{Until reaching root}
        \State $\mathcal{L}' \gets [\;]$ \Comment{placeholder of the next level.}
        \For{$\textit{siblings } (v_1, v_2) \textit{ in } \mathcal{L}$ with parent $u$}
            \State $X_u' \gets X_{v_1}\cup X_{v_2}$\Comment{Merge each pair of siblings}
            \State $\Pi_f^{(u)} \gets$ a fair matching constructed on $X_u'$
            \State $X_u\gets \hp(X_u',\ranges_{|X_u'},\Pi_f^{(u)})$
           % \State $M \gets \textit{FairHalf(M)}$\Comment{The fair deterministic Halving.}
            \State $\mathcal{L}'.append(u)$
        \EndFor
        \State $\mathcal{L} \gets \mathcal{L}'$
    \EndWhile
    \State $\root \gets \mathcal{L}[0]$
    \State $\widehat{\X}\gets X_{\root}$
    \While{$|\widehat{\X}| > c_0 \frac{4d}{\eps^2} \log \frac{2d}{\eps}$}\Comment{Until finding $\frac{\eps}{2}$-sample}
        \State Construct fair matching $\Pi_f$ on $\widehat{\X}$
        \State $\widehat{\X}\gets \hp(\widehat{\X},\ranges_{|\widehat{\X}},\Pi_f)$
    \EndWhile
    \State $\X \gets \widehat{\X}$\Comment{The $\frac{\eps}{2}$-sample.}
    \State $\res\gets$ Execute algorithm from Theorem~\ref{thm:discrepancy} on $(\X,\ranges_{|\X})$
    \State \Return $\res$
\EndFunction
%\Function{FairDPGreedy}{$\points, \ranges, \eps$} $\gets$ Returns a fair $\eps$-net by greedily picking the best $V \subseteq \points$ at each step. $V$ is the minimal subset, maintaining the ratio of colors (See~\cite{dehghankar2024fair}).\EndFunction
\caption{Fair Sketch-and-Merge algorithm (FSM) for Finding an $\eps$-net satisfying DP fairness.}
\label{alg:epsnet:dp:sketch}
\end{algorithmic}
\end{algorithm}

\paragraph{Algorithm}
Let $p$ be a small parameter value that is set later.
Let $\parti$ be the partitioning of $X$ into $\frac{n}{2^p}$ equally-sized fair subsets. Each partition $P\in \parti$ contains $2^p$ points of the same color $\group\in\groups$.
We construct an empty binary partition tree with $\frac{n}{2^p}$ leaf nodes.
For a node $v$ of the partition tree, let $X_v\subseteq X$ be the points stored in node $v$. Initially, $X_v=\emptyset$ for every node $v$ in the partition tree.
Each partition $P\in \parti$ is assigned to a leaf node $v_P$ in the partition tree and $X_{v_P}=P$.
Let $\mathcal{L}_j$ be the set of nodes in the $j$-th level of the tree.
Initially, $\mathcal{L}_0=\{v_P\mid P\in \parti\}$ be the set of leaf nodes of the partition tree.
We repeat the following until we reach the root of the tree. For every pair of sibling nodes $v_1,v_2$ in $\mathcal{L}_j$ with parent $u\in\mathcal{L}_{j+1}$ we define $X_u'=X_{v_1}\cup X_{v_2}$ (Merging step). We construct a fair matching $\Pi_f^{(u)}$ on $X_u'$ as we had in the algorithm of \S~\ref{subsec:discr1}, and we execute $\hp(X_u',\ranges_{|X_u'},\Pi_f^{(u)})$. Let $X_u$ be the set of points from $X_u'$ with color $+1$ returned by the halving procedure (Halving step).
Let $\root$ be the root of the partition tree and $X^{(\root)}$ the points stored in it following the procedure above.
We set $\widehat{\X}=X_\root$.
While $|\widehat{\X}|>c_0\frac{4d}{\eps^2}\log(\frac{2d}{\eps})$, we repeat the following recursive approach. We construct a fair matching $\Pi_f$ on $\widehat{\X}$ and we execute $\hp(\widehat{\X},\ranges_{|\widehat{\X}},\Pi_f)$. Let $\widehat{\X}$ be the set of points with color $+1$ returned by the halving procedure.
The first time that we find $|\widehat{\X}|<c_0\frac{4d}\log\frac{2d}{\eps}$, we stop the repetitions and we set $\X=\widehat{\X}$. 
In the end, we run the algorithm from Theorem~\ref{thm:discrepancy} on the range space $(\X,\ranges_{|\X})$ with parameter $\eps/2$, and let $\res$ be the returned set of points. We return $\res$.
We show the pseudocode of our algorithm in Algorithm~\ref{alg:epsnet:dp:sketch}.
The overall structure of our algorithm using a partition is shown in Figure~\ref{fig:partitiontree}.

\paragraph{Correctness and runtime analysis}
\begin{lemma}
\label{lem:fairsample}
    The set $\X$ is a fair $\frac{\eps}{2}$-sample of $\rangespace$.
\end{lemma}
\vspace{-4mm}
\begin{proof}
First, we argue that $\X$ is an $\eps/2$-sample. Indeed, using an arbitrary initial partitioning $\parti$ and arbitrary matchings in each node of the partition tree, it is known that in the root of the partition tree, an $\frac{\eps}{2}$-sample is found~\cite{chazelle2000discrepancy}. Two main arguments are used to make this claim. First, if $X'$ is the set of points with color $+1$ returned by the Halving procedure in Lemma~\ref{lem:halving2}, then $X'$ is a $\frac{2\Delta}{n}$-sample of size $n/2$. Since Lemma~\ref{lem:halving2} works for an arbitrary input matching, the claim is still valid using our constructed fair matchings. Second,  if  $\epssample_1$ and $\epssample_2$ are two equal-size $\eps$-samples for the range spaces $(\points_1, \ranges_{|\points_1})$ and $(\points_2, \ranges_{|\points_2})$, respectively, where $\points_1$ and $\points_2$ are disjoint subsets of $\points$ with the same size, then $\epssample_1 \cup \epssample_2$ is an $\eps$-sample of $\points_1 \cup \points_2$. %We can straightforwardly claim the same argument using  the notion of fairness:  if $\epssample_1$ and $\epssample_2$ are two equal-size fair $\eps$-samples for the range spaces $(\points_1, \ranges_{|\points_1})$ and $(\points_2, \ranges_{|\points_2})$, respectively, where $\points_1$ and $\points_2$ are disjoint subsets of $\points$ with the same size, then $\epssample_1 \cup \epssample_2$ is a fair $\eps$-sample of $\points_1 \cup \points_2$.
Hence, $\X$ is an $\eps/2$-sample.
Let $X^{(j)}=\bigcup_{v\in \mathcal{L}_j}X_v$ be the union of all points stored in the nodes of the $j$-th level of the partition tree. Since we start with a fair partitioning and in each node we compute a fair matching we have that $\frac{|X^{(j)}_\group|}{|X^{(j)}|}=\frac{\frac{1}{2}|X^{(j-1)}_\group|}{\frac{1}{2}|X^{(j)}|}=\frac{|X_\group|}{|X|}$, for every $\group\in \groups$. Hence we have that $\frac{|\X_\group|}{|\X|}=\frac{|X_\group|}{|X|}$, for every $\group\in \groups$ and $\X$ is a fair $\eps/2$-sample of $\rangespace$.
\end{proof}

\begin{lemma}
The returned set $\res$ is a fair $\eps$-net of $\rangespace$.
\end{lemma}
\begin{proof}
   It is known from~\cite{chazelle2000discrepancy} that 
   an $\eps_1$-net of a $\eps_2$-sample of $\rangespace$, is an $(\eps_1+\eps_2)$-net of $\rangespace$. Similarly, it is straightforward to see that a fair $\eps_1$-net of a fair $\eps_2$-sample of $\rangespace$, is a fair $(\eps_1+\eps_2)$-net of $\rangespace$. From Lemma~\ref{lem:fairsample}, we know that $\X$ is a fair $\frac{\eps}{2}$-sample of $\rangespace$. The algorithm from Theorem~\ref{thm:discrepancy} is executed on $(\X,\ranges_{|\X})$ with parameter $\frac{\eps}{2}$, and it returns the set $\res$ which is a fair $\frac{\eps}{2}$-net of $(\X,\ranges_{|\X})$. Hence, $\res$ is also a fair $\eps$-net of $\rangespace$.
\end{proof}

\begin{lemma}
    The running time of the algorithm is $O\left(n\cdot\frac{d^{3d}}{\eps^{2d}}\log^d\left(\frac{d}{\eps}\right)\right)$.
\end{lemma}
\begin{proof}
The fair partitioning is constructed in $O(n)$ time.
Each Halving procedure on  a node $u$ of the partition tree runs in $O(|X_u|\cdot\min\{|\ranges_{|X_u}|,|X_u|^d\})$. Furthermore, every fair matching $\Pi_f^{(u)}$ is constructed in $O(|X_u'|)$ time.
We have $n / 2^p$ partitions, where each of them requires an execution of the Halving procedure.
As a result, all Halving procedures in the leaf level of the partition tree take $O(\frac{n}{2^p}(2^p)^{d+1})$ time.
Any other following steps take asymptotically smaller running time following the technical details~\cite{chazelle2000discrepancy}, which we ignore here. By a carefully chosen value of $p$, we construct the fair $\eps/2$-sample $\X$ in $O\left(d^{3d} (\frac{1}{\eps^2}\log \frac{d}{\eps})^d n\right)$ time.
Notice that $|\X|=O(\frac{d}{\eps^2}\log\frac{d}{\eps})$. Since the VC dimension is $d$ we have that $|\ranges_{|\X}|=O((\frac{d}{\eps^{2}}\log\frac{d}{\eps})^d)$. Hence, the algorithm from Theorem~\ref{thm:discrepancy} on $\X$ runs in $O((\frac{d}{\eps^{2}}\log\frac{d}{\eps})^{d+1}\log\frac{d}{\eps})$ time. Overall, our algorithm runs in

\vspace{-6mm}
\[O\left(n\cdot \frac{d^{3d}}{\eps^{2d}}\log^d \frac{d}{\eps} + \frac{d^{d+1}}{\eps^{2d+2}}\log^{d+2}\frac{d}{\eps}\right)\]
time.
We expect that $n\gg \frac{1}{\eps^2}\log^{2}\frac{d}{\eps}$ so the first term dominates the running time. The lemma follows.
\end{proof}

Putting everything together, we conclude with Theorem~\ref{thm:discrepancy2}.
\begin{theorem}
\label{thm:discrepancy2}
    Given a range space $\rangespace$, with $|X|=n$, VC dimension $d$, DP constraints $\mathcal{T}$, and a parameter $\eps\in(0,1)$,
    there exists a deterministic algorithm that constructs a fair $\eps$-net with respect to $\mathcal{T}$, of size $O\left(\frac{d}{\eps}\log\frac{d}{\eps}\right)$ in $O\left(n\cdot\frac{d^{3d}}{\eps^{2d}}\log^d\left(\frac{d}{\eps}\right)\right)$ time.
\end{theorem}

\vspace{-2.5mm}
\subsection{Satisfying Custom-Ratio Fairness}\label{sec:custom}\label{sec:epsnets:cr}
In this section, we discuss the second, more general, fairness definition where the user provides the ratios $\mathcal{T} = \{\tau_1, \tau_2, \cdots, \tau_k\}$ and requires the returned $\eps$-net $\res$ to satisfy
$\frac{|\res_{\group_i}|}{|\res|}=\tau_i$ for every $\group_i\in\groups$. 

Through the whole subsection, we assume that we have access to the ranges $\ranges$; in other words, the ranges are explicitly given to us as input. This assumption is valid in most cases; for example, the discrepancy approach to deterministically find an $\eps$-net is based on this assumption~\cite{har2011geometric, chazelle2000discrepancy}.

Assuming that we have access to the ranges $\ranges$ and $|\ranges|=m$, finding an $\eps$-net of the range space $\rangespace$ can be reduced to the \emph{hitting set} problem. In other words, based on the Definition~\ref{def:epsnet}, $\epsnet$ is an $\eps$-net of $\rangespace$, if and only if it hits all the ranges in $\ranges_{\eps} = \{\range \mid \range \in \ranges, \; \frac{|\range|}{|\points|} \geq \eps\}$.

Given the range space $\rangespace$, one can first filter out the light ranges and keep only the {\it heavy} ranges, aka $\ranges_{\eps}$. Then, compute a hitting set for the new range space $(\points, \ranges_{\eps})$. 
% \marginpar{\scriptsize \color{blue} Add a pseudo-code \color{black}}
%The hitting set problem is the dual of the Set Cover problem in the dual range space. As a result, a
Any algorithm solving the hitting set problem would result in solving the $\eps$-net problem.

\vspace{-4mm}
\[
\text{$\varepsilon$-net}
\;\overset{\text{reduction}}{\underset{\text{over } \ranges_{\eps}}{\longrightarrow}}\;
\text{Hitting Set}
%\;\overset{\text{reduction}}{\underset{\text{dual space}}{\longrightarrow}}\;
%\text{Set Cover}
\]
\vspace{-2mm}

\newcommand{\poly}{\mathsf{poly}}
In the fairness context, any fair hitting set given the constraints $\mathcal{T}$ is also a fair $\eps$-net. Solving the Fair Hitting Set problem is discussed in~\cite{dehghankar2024fair} in the general setting. They provide a randomized $O(\log n)$-approximation algorithm for the fair set cover that runs in $O(\poly(n,m))$ time. The approximation factor holds in expectation. The algorithm in~\cite{dehghankar2024fair} is a greedy algorithm that picks the most promising points at each step. As a result, by applying the same algorithm, we can compute a {\bf fair} $\eps$-net for the range space $\rangespace$ with expected size $O(\opt\cdot\log n)$, where $\opt$ is the fair $\eps$-net with the minimum size. We note that this algorithm does not find a fair $\eps$-net with an absolute size as we had in Theorems~\ref{thm:sampling}, \ref{thm:discrepancy}, and~\ref{thm:discrepancy2}. Instead, our algorithm for custom-ratios returns a fair $\eps$-net with size larger by an (expected) $O(\log n)$ factor from the smallest possible fair $\eps$-net with respect to the CR constraints $\mathcal{T}$.
\begin{theorem}
\label{thm:CR}
    Given a range space $\rangespace$, with $|X|=n$, $|\ranges|=m$, CR constraints $\mathcal{T}$, and a parameter $\eps\in(0,1)$,
    there exists a randomized algorithm that constructs a fair $\eps$-net with respect to $\mathcal{T}$, of expected size $O(\opt\cdot \log(n))$ in $O(\poly(n,m))$ time, where $\opt$ is the size of the smallest fair $\eps$-net with respect to $\mathcal{T}$.
\end{theorem}

The algorithm from Theorem~\ref{thm:CR} works in general range spaces; however, assuming a range space $\rangespace$ with a bounded VC-dimension would result in better approximation algorithms~\cite{lit88, fs97, ch09, cch09}. The Fair Geometric Hitting Set  (FGHS) problem is discussed in detail in \S~\ref{sec:fghs}. Here, we only use the final result we get for the FGHS problem and use it to derive an approximation algorithm for constructing a fair $\eps$-net satisfying CR constraints $\mathcal{T}$.
In Theorem~\ref{thm:FGHS}, we show an $O\left(\!\max\{\log\frac{1}{\prob},d\log(\opt_{FGHS})\}\!\cdot\!\log(\frac{k}{\prob})\right)$-approximation algorithm (with probability at least $1-\prob$) for the Fair Geometric Hitting Set problem (on CR constraints), where $\opt_{\mathsf{FGHS}}$ is the optimum solution. By leveraging the algorithm from Theorem~\ref{thm:FGHS} on $(\points, \ranges_{\eps})$, we return a fair $\eps$-net $\res$ with respect to CR constraints $\mathcal{T}$ of size $O\left(\!\max\{\log\frac{1}{\prob},d\log(\opt)\}\!\cdot\!\log(\frac{k}{\prob})\!\cdot\!\opt\right)$ with probability at least $1-\prob$. The algorithm's time complexity is $O(n\cdot m +\mathcal{M}(n,m+k))$, where $\mathcal{M}(n,m+k)$ is the running time to solve a linear program with $O(n)$ variables and $O(m+k)$ constraints (Theorem~\ref{thm:FGHS}).
%Here, $|\epsnet| = \Theta(\frac{d}{\eps} \log \frac{d}{\eps})$, is the size of an $\eps$-net found by random sampling.
% \marginpar{\scriptsize \color{blue} Talk about corner cases; not all ratios are satisfiable, like Fair Set Cover. \color{black}}

\begin{theorem}
\label{thm:CR2}
    Given a range space $\rangespace$, with $|X|=n$, $|\ranges|=m$, CR constraints $\mathcal{T}$, and a parameter $\eps\in(0,1)$,
    there exists a randomized algorithm that constructs a fair $\eps$-net with respect to $\mathcal{T}$, of size $O\left(\!\max\{\log\frac{1}{\prob},d\log(\opt)\}\!\cdot\!\log(\frac{k}{\prob})\!\cdot\!\opt\right)$ with probability at least $1-\prob$, in $O(n\cdot m +\mathcal{M}(n,m+k))$ time, where $\opt$ is the size of the smallest fair $\eps$-net with respect to $\mathcal{T}$, $k$ is the number of colors in $X$, and $\mathcal{M}(n,m+k)$ is the running time to solve a linear program with $O(n)$ variables and $O(m+k)$ constraints.
\end{theorem}

%\begin{theorem}
%    The above algorithm returns a fair $\eps$-net $\res$ satisfying the CFC constraint, with a size of $O(|\epsnet| \log |\epsnet| \log k)$. The running time is the same as the running time in Theorem~\ref{thm:fgsc}, plus an additional $\Theta(m)$ time for enumerating the ranges in $\ranges$ and building $\ranges_{\eps}$.
%\end{theorem}

It is worth mentioning that in \S~\ref{sec:fghs}, the Fair Geometric Hitting Set is solved by a reduction to the weighted fair $\eps$-net with DP constraints. As a result, the whole algorithm reduces an instance of the fair $\eps$-net with CR constraint to an instance of the weighted fair $\eps$-net problem with DP constraints.

\vspace{-4mm}
\[
\text{Fair $\eps$-net (CR)}
\;\overset{\text{red.}}{\underset{}{\longrightarrow}}\;
\text{FGHS (CR)}
\;\overset{\text{red.}}{\underset{}{\longrightarrow}}\;
\text{Weighted Fair $\eps$-net (DP)}
\]
\vspace{-4mm}
\vspace{-2.5mm}
\subsection{Fair $\eps$-sample}\label{sec:epssample}
The fair $\eps$-sample problem under both demographic parity (DP) and custom-ratio (CR) constraints is addressed by extending the techniques introduced for $\eps$-nets. For DP, fair $\eps$-samples can be constructed following similar principles as in Section~\ref{sec:epsnets}. However, satisfying arbitrary CR constraints is not always feasible. In particular, if a range in the space separates all members of a demographic group, then any fair $\eps$-sample must preserve the group's proportion within an $\eps$ margin, limiting the range of achievable target ratios. We present a formal argument and provide a counterexample illustrating this limitation in Appendix~\ref{app:epssample}.
\vspace{-3mm}
\section{Fair Geometric Hitting Set}\label{sec:fghs}
\vspace{-0.6mm}
In this section, we give an algorithm for satisfying the CR constraints in the Geometric Hitting Set. Given the ratios $\mathcal{T} = (\tau_1, \tau_2, \cdots, \tau_k)$ the goal of FGHS is to find the smallest subset $\res^* \subseteq X$ such that:

\vspace{-.5mm}
\begin{enumerate}
    \item $\forall \group_\ell\in \groups,\; \frac{|\res^*_{\group_\ell}|}{|\res^*|} = \tau_l$,
    \item $\forall R\in\ranges,\; R\cap \res^*\neq \emptyset$.
\end{enumerate}
\vspace{-1mm}
%Here, the colors are defined on ranges, so $\ranges^*_{\group_\ell}$ represents family of sets in $\ranges^*$ with color $\group_\ell$. 
Equivalently, one can consider Fair Geometric Set Cover defined in the dual range space where sets are associated with colors.

There are two popular algorithms for solving the (unfair) Geometric Hitting Set problem. The first approach relies on the Multiplicative Weight Update method~\cite{arora2012multiplicative, lit88, fs97}. This is a greedy algorithm that starts by choosing a subset of points. If this subset is not a hitting set for $\ranges$, they re-weight all the points that belong to at least a range that is not hit by the current solution. Iteratively repeating this process results in finding a $\log \opt_{\mathsf{GHS}}$ approximation solution, where $\opt_{\mathsf{GHS}}$ is the optimum geometric hitting set~\cite{har2011geometric}.

The second approach formulates the problem as an Integer Program, and solves the relaxed Linear Program~\cite{lon01, cch09, ch09, har2011geometric}. Our algorithm is a variant of this approach that adds the fairness constraints inside the LP formulation. Here, we assume that we have access to all ranges $\ranges$ with $|\ranges|=m$.
%; in other words, we have access to all points in the dual range space (we are mostly concerned with the {\bf dual range space}, which has a bounded VC dimension $d$).

The IP for solving the FGHS problem is as follows:
\vspace{-.5em}
\begin{small}
\begin{align}\label{eq:ip}
    \text{min} \quad & \sum_{p_i\in X} z_i \\
    \text{s.t.} \quad 
     \sum_{p_i\in R} z_i &\geq 1, \quad \forall R\in \ranges\notag\\
     \sum_{p_i\in X_{\group_\ell}}z_i&=\tau_\ell\cdot \sum_{p_i\in X}z_i, \quad\forall\group_\ell\in\groups\notag\\
    z_i &\in \{0, 1\}, \forall i\in[n]\notag
\end{align}
% \vspace{-3mm}
\end{small}

In this formulation, $z_i$ is a variable that indicates whether the point $p_i\in X$ is selected in the final cover. The first type of constraints ensures that all ranges are hit. The second type of constraints ensures the CR fairness constraints. The equivalent LP formulation considers the real domain $[0, 1]$ for variables $z_i$. By using a variable $f = \sum_i z_i$, we can formulate the LP as:
\vspace{-0.6em}
\begin{small}
\begin{align}\label{eq:mlp}
    \text{min} \quad & f \\
    \text{s.t.} \quad 
     \sum_{p_i\in X} z_i &= f \notag \\
     \sum_{p_i\in R} z_i &\geq 1, \quad \forall R\in\ranges \notag\\
    \sum_{p_i\in X_{\group_\ell}}z_i&=\tau_\ell\cdot f, \quad\forall\group_\ell\in\groups\notag\\
    z_i\in[0,1],& \forall i\in[n],\quad f\geq 0\notag
\end{align}
\end{small}

Now, define $\bar{w}_i = \frac{z_i}{f}$ and $\bar{\eps} = \frac{1}{f}$. We can rewrite the LP:
\vspace{-0.5em}
\begin{small}
\begin{align}\label{eq:lp}
    \text{max} \quad & \bar{\eps} \\
    \text{s.t.} \quad 
     \sum_{p_i\in X} \bar{w}_i &= 1 \notag \\
     \sum_{p_i\in R} \bar{w}_i &\geq \bar{\eps}, \quad \forall R\in\ranges \notag\\
    \sum_{p_i\in X_{\group_\ell}}\bar{w}_i&=\tau_\ell,\quad\forall\group_\ell\in\groups\notag\\
    \bar{w}_i\in[0,1], &\forall i\in[n], \quad\bar{\eps}\geq 0\notag
\end{align}
\end{small}
Interestingly, the LP in~\eqref{eq:lp} is interpreted as follows. The LP sets some weights to the points such that every range in $\ranges$ has a total weight of at least $\bar{\eps}$, and the total weight of the points with color $\group_\ell$ is $\tau_\ell$. Intuitively, a weighted fair $\bar{\eps}$-net as defined in~\ref{subsec:weightedfairepsnet}.
Using these observations, we are ready to describe our algorithm for the FGHS problem.

\paragraph{Algorithm}
Construct the instance of the LP~\eqref{eq:lp}. Use an LP solver to solve the LP ~\eqref{eq:lp}. Let $\eps$ be the value of variable $\bar{\eps}$ in the optimum LP solution and $w_i$ be the value of the variable $\bar{w}_i$ for every $p_i\in X$ in the optimum LP solution. Let $\vec{w}=\{w_1,\ldots, w_n\}$ be the vector of all points' weights.
We run the algorithm from Theorem~\ref{thm:weightedsampling} on $\rangespace$ with DP constraints $\mathcal{T}$, parameter $\eps$, assuming that each point $p_i\in X$ has weight $w_i$. Let $\res$ be the weighted fair $\eps$-net returned by the algorithm from Theorem~\ref{thm:weightedsampling}. We return $\res$.
The pseudocode is shown in Algorithm~\ref{alg:fghs}.

\paragraph{Correctness and runtime analysis}
Let $\opt_{FGHS}$ be the size of the optimum solution for the FGHS problem on $\rangespace$.
\begin{lemma}
    The algorithm returns a hitting set of size 
    $$O\left(\max\{\log\frac{1}{\prob},d\log(\opt_{FGHS})\}\cdot\log(\frac{k}{\prob})\cdot \opt_{FGHS}\right)$$
    %$$O\left(\frac{1}{\mu_\mathcal{T}}\max\{\log\frac{1}{\prob},d\log(\opt_{FGHS})\}\cdot\log(\frac{k}{\prob})\cdot \opt_{FGHS}\right)$$
    with probability at least $1-\prob$.
\end{lemma}
\begin{proof}
We first observe that solving the LP relaxation of IP~\eqref{eq:ip} or any of the LP listed above is equivalent (i.e., given an optimal solution to one, we can derive an optimal solution to any other).
Let $z_i^*$ be the value of the variable $z_i$ in the optimal solution of IP~\eqref{eq:ip}. Notice that $\opt_{FGHS}=\sum_{p_i\in X}z_i^*$.
Furthermore, let $f^*$ be the value of the variable $f$ in the optimal solution of LP~\eqref{eq:mlp}.
Since LP~\eqref{eq:mlp} is equivalent to the LP relaxation of IP~\eqref{eq:ip} it holds that $\opt_{FGHS}\geq f^*$.
By definition, it also holds that $f^*=\frac{1}{\eps}$. Hence,
$$\opt_{FGHS}=\sum_{p_i\in X}z_i^*\geq f^*=\frac{1}{\eps}\Leftrightarrow \frac{1}{\eps}\leq \opt_{FGHS}.$$
The algorithm from Theorem~\ref{thm:weightedsampling} executed on $\rangespace$ with parameters $\eps, \prob$ and weights $\vec{w}=\{w_1,\ldots, w_n\}$ such that $W_{\group_\ell}=\sum_{p_i\in X_{\group_\ell}}w_i=\tau_\ell$ for every $\ell\in[k]$, returns a weighted fair $\eps$-net $\res$ of size $$O\!\left(\!\frac{1}{\eps}\max\{\log\frac{1}{\prob},d\log\frac{1}{\eps}\}\!\!\cdot\!\log\frac{k}{\prob}\!\right).$$
Since $\frac{1}{\eps}\leq \opt_{FGHS}$ the result of the lemma follows.
\end{proof}

Let $\mathcal{M}(x,y)$ be the running time to solve an LP with $O(x)$ variables and $O(y)$ constraints.
\begin{lemma}
    Our algorithm has a time complexity of $O(n\cdot m + \mathcal{M}(n,m+k))$.
\end{lemma}
\vspace{-2mm}
\begin{proof}
    We construct the instance of the LP~\eqref{eq:lp} in $O(n\cdot m)$ time, traversing all points inside every range in $\ranges$. The constructed LP has $O(n)$ variables and $O(m+k)$ constraints so the LP solver runs in $O(\mathcal{M}(n,m+k))$ time. Finally, the algorithm from Theorem~\ref{thm:weightedsampling} runs in linear time with respect to the input points. The lemma follows.
\end{proof}

Putting everything together, we conclude with the next theorem.
\begin{theorem}
    \label{thm:FGHS}
    Given a range space $\rangespace$ with $|X|=n$, $|\ranges|=m$, CR constraints $\mathcal{T}$, and a parameter $\prob\in(0,1)$, there exists a randomized algorithm that constructs a fair geometric hitting set of $\rangespace$, of $O\left(\!\max\{\log\frac{1}{\prob},d\log(\opt_{FGHS})\}\!\cdot\!\log(\frac{k}{\prob})\!\cdot\! \opt_{FGHS}\!\right)$ size with probability at least $1-\prob$, in $O(n\cdot m + \mathcal{M}(n,m+k))$ time, where $\opt_{FGHS}$ is the size of the optimum FGHS solution,  $k$ is the number of different colors in set $X$, and $\mathcal{M}(n,m+k)$ is the running time to solve an LP with $O(n)$ variables and $O(m+k)$ constrains.
\end{theorem}

\begin{algorithm}[bt]
\begin{algorithmic}[1]
\small
\Require Range space $\rangespace$, set of custom ratios $\mathcal{T}$.
\Ensure The Fair Geometric Hitting Set $\res$.
\Function{FGLP}{$\points, \ranges, \mathcal{T}$}
    \State $\eps, \vec{w} \gets \mathsf{LPSolver}(\points, \ranges, \mathcal{T})$
    \State $\res \gets$ Execute algorithm from Theorem~\ref{thm:weightedsampling} on $\rangespace$ with parameter $\eps$ and weights $\vec{w}$
    \State \Return $\res$
\EndFunction
\Function{$\mathsf{LPSolver}$}{$\points, \ranges, \mathcal{T}$}
    \State \textit{Solves the LP in Equation~\ref{eq:lp}.}
    \State \Return \textit{$\eps$ and the weights $\vec{w} =\{w_1, \cdots, w_n\}$.}
\EndFunction
\caption{Algorithm for finding a Fair Geometric Hitting Set through LP relaxation (FGLP).}
\label{alg:fghs}
\end{algorithmic}
\end{algorithm}
\vspace{-2.5mm}
\section{Related Work and Applications}\label{sec:related}

Fairness in data-driven systems has become a central topic in Computer Science, particularly in Machine Learning~\cite{barocas2023fairness,mehrabi2021survey}. 
However, relatively less attention has been given to algorithmic fairness in the context of combinatorial optimization problems~\cite{wang2022balancing}.
Related work includes satisfying fairness constraints while optimizing sub-modular functions \cite{wang2022balancing}, and satisfying fairness in Coverage Maximization problems ~\cite{asudeh2023maximizing,bandyapadhyay2021fair}, Set Cover~\cite{dehghankar2024fair}, Hitting Sets~\cite{inamdar2023fixed}, and Facility Locations~\cite{jung2019center}.
Fairness considerations have also been extended to a variety of other algorithmic domains, in related problems like Matching \cite{garcia2020fair,esmaeili2023rawlsian}, 
Resource Allocation~\cite{mashiat2022trade}, Ranking~\cite{asudeh2019designing,zehlike2017fa}, and Clustering~\cite{makarychev2021approximation,thejaswi2021diversity}.
Beyond these, fairness has been investigated in approximate data processing (AQP) data structures, including data-informed hashmaps~\cite{shahbazi2024fairhash} and locality-sensitive hashing~\cite{aumuller2021fair}.

% \color{blue}
For instance, fairness in coverage has been enforced by using coloring constraints to minimize discrepancies in the number of elements of different colors covered by the selected sets~\cite{asudeh2023maximizing}. Similarly, fairness has been studied in the context of vertex cover and edge cover problems~\cite{bandyapadhyay2021fair}, as well as in the set cover problem where colors are assigned to the sets rather than the elements~\cite{dehghankar2024fair}. However, none of these settings consider the geometric case where the range space has a bounded VC-dimension.

In contrast, this work focuses on the hitting set problem with fairness constraints applied directly to the elements (i.e., points) in geometric settings. While some of the existing methods could be adapted to the dual of this problem and serve as baselines, the setting considered here is fundamentally different: the focus is on the geometric case, where stronger approximation guarantees are achievable compared to general set cover or hitting set formulations. Additionally, a variant of the hitting set problem has been proposed in which the goal is to avoid selecting too many elements of the same type~\cite{inamdar2023fixed}, which is different from the objective considered here.
% \color{black}

To the best of our knowledge, this is the first paper to study fairness in geometric approximation algorithms, particularly for \enet, $\eps$-sample, and geometric hitting set problems.

\vspace{-2.5mm}
\subsection{\enet 's application in Data-driven Systems}\label{sec:application}
In Example~\ref{ex:enet}, we highlighted the application of \enets as {\bf dataset summaries for range queries}, where the goal is to compress potentially very large datasets to a small subset that contains at least one tuple from any possible (axis-parallel) range query. The \enet can then facilitate AQP by quickly identifying a tuple in the set that satisfies a query range.
In the following, we briefly outline some other applications of \enets in data-driven systems:

\vspace{-1mm}
\stitle{$k$-Nearest Neighbor (kNN) search}
kNN search is popular for settings such as geospatial data, where one would like to retrieve the $k$ nearest tuples with minimum Euclidean distance to a given query point $q$~\cite{roussopoulos1995nearest,dhanabal2011review,yao2010k}. In such cases, the queries can be viewed as $d$-balls ($d$-dimensional hyperspheres) centered at $q$. Given a large dataset of tuples, an \enet provides a small subset that guarantees to hit the top-$k$ (for the corresponding $\eps$) of any given $d$-ball. Hence, the NN of the \enet to $q$ belongs to the $k$NN of the complete dataset.

\vspace{-1mm}
\stitle{Regret-minimizing Sets}
Another application of \enets is in regret-minimizing representative sets of a dataset, ~\cite{nanongkai2010regret}: compact sets that minimize a notion of regret in approximate top-$k$ query processing.
In particular, rank-regret representatives are compact subsets that guarantee to contain at least one of the top-$k$ results according to any linear ranking function~\cite{asudeh2019rrr}. Hence, by defining the range space as the universe of half-spaces, and setting $\eps=\frac{k}{n}$, an \enet for a dataset is a rank-regret representative~\cite{asudeh2022finding}.

\vspace{-1mm}
\stitle{Application demonstration in machine learning}
\enets, geometric hitting Sets, and set covers have diverse applications across machine learning tasks such as sample-efficient learning, active learning, clustering, and interoperability. For example, \enets offer a principled way to select small informative subsets for active learning under geometric and distributional assumptions, with applications in learning linear separators ~\cite{balcan2013active}, convex bodies ~\cite{har2021active}, and online settings ~\cite{bhore2024online}. In clustering, coresets based on \enets and hitting sets enable efficient approximations for $k$-means and $k$-median objectives \cite{har2004coresets, feldman2011unified}.
Hitting set formulations support interpretable and sparse models like the Set Covering Machines \cite{marchand2003set, hussain2004linear} and predictive checklists \cite{zhang2021learning}, where minimal feature sets capture key behaviors. They also find use in motion planning \cite{shaw2024towards} and sensor network coverage \cite{wu2020optimal}.

\vspace{-2.5mm}
\section{Experiments}\label{sec:exp}
In this section, we evaluate the proposed algorithms on (a) real-world datasets from related applications and (b) synthetic data. The code and other artifacts are publicly available\footnote{\color{blue}\href{https://github.com/UIC-InDeXLab/FairNet}{Github repository}\color{black}}. Appendix~\ref{app:exp} contains a detailed discussion of the experiments and additional results.

\begin{figure*}[ht]
\centering
    \begin{subfigure}[t]{0.32\linewidth}
        \includegraphics[width=.95\linewidth]{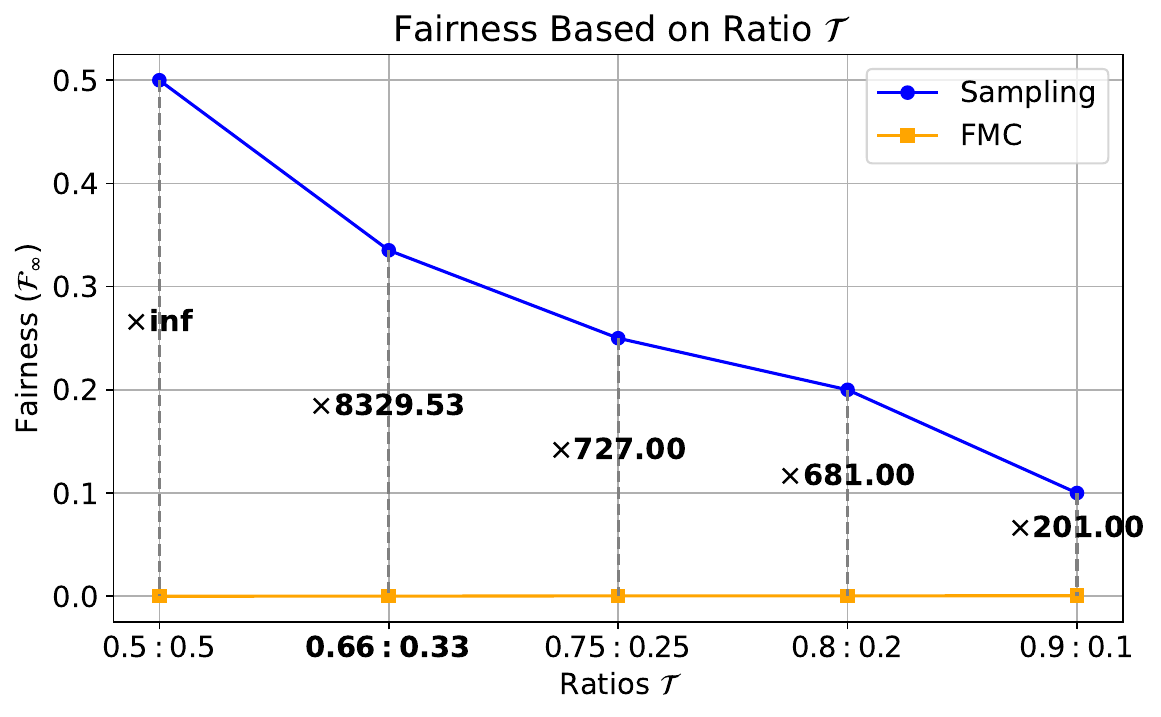}
        \vspace{-2.5mm}\caption{Adults Dataset}
        % \label{fig:cover_size_resume}
    \end{subfigure}
    \hfill
    \begin{subfigure}[t]{0.32\linewidth}
        \includegraphics[width=.95\linewidth]{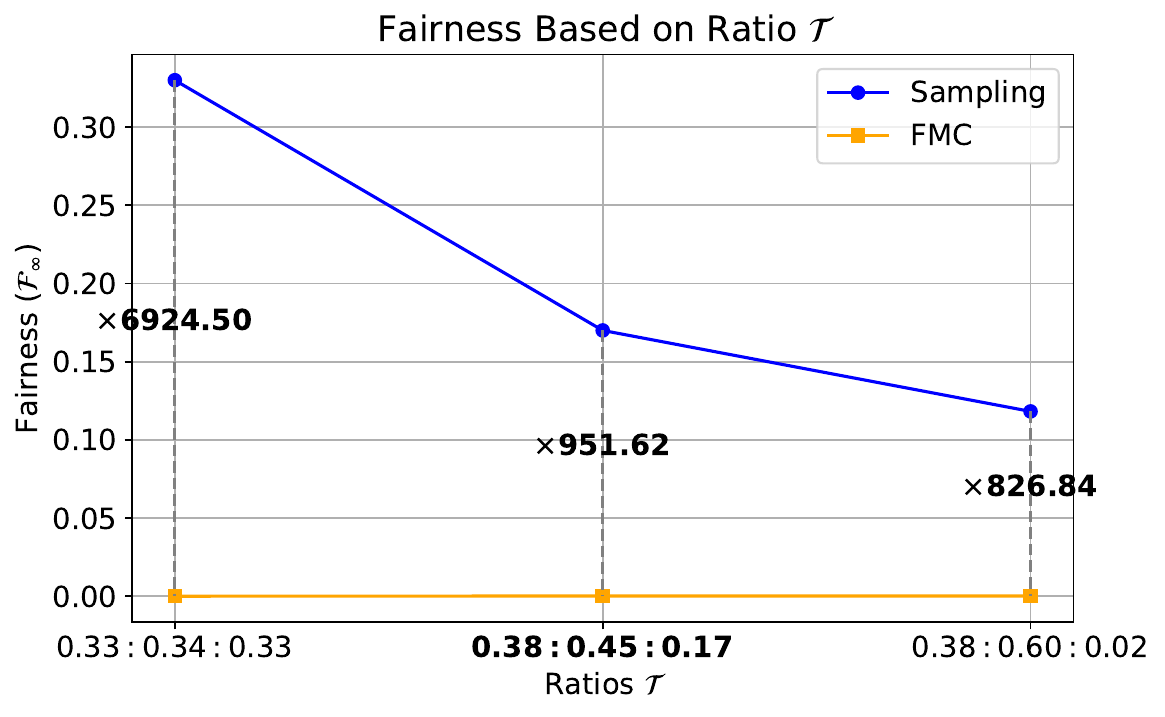}
        \vspace{-2.5mm}\caption{COMPAS Dataset}
        % \label{fig:fairness_resume}
    \end{subfigure}
    \hfill
    \begin{subfigure}[t]{0.32\linewidth}
        \includegraphics[width=.95\linewidth]{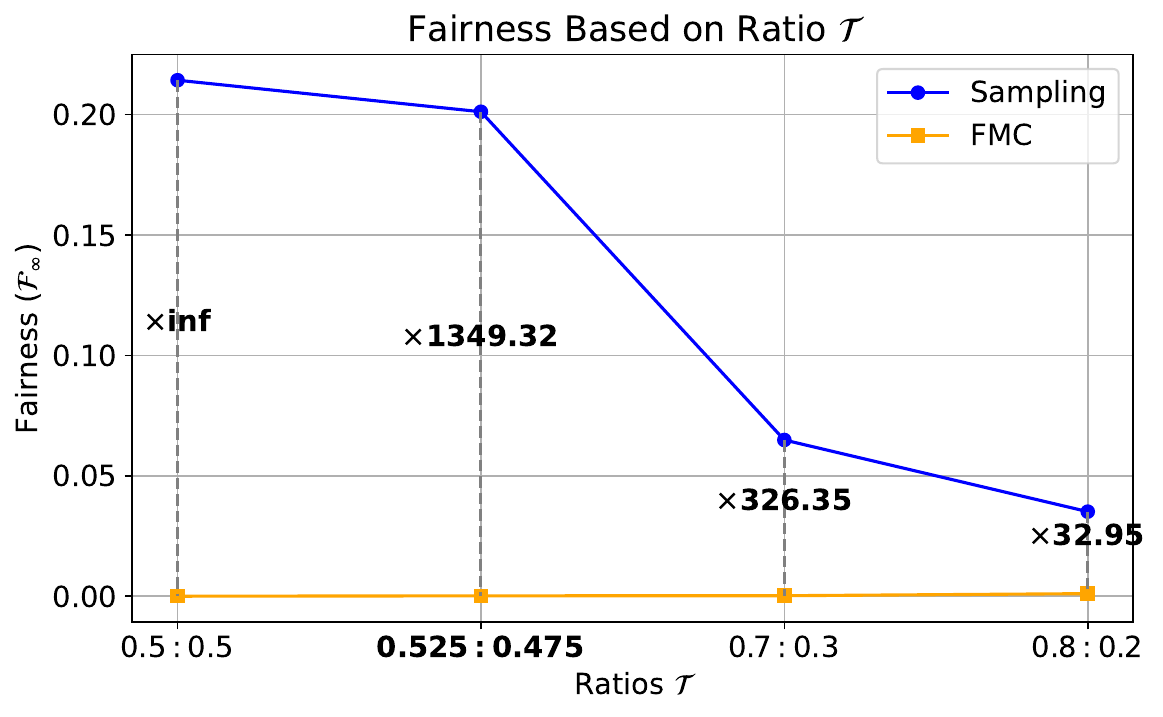}
        \vspace{-2.5mm}\caption{College Admission Dataset}
        % \label{fig:running_time_resume}
    \end{subfigure}
\vspace{-4mm}\caption{Measuring the fairness of the output of Fair Monte-Carlo (FMC) and standard sampling for dataset summarization. The fairness constraint is CR with its ratios $\mathcal{T}$ given on the x-axis.}
\label{fig:rds:fairness:cr}
\vspace{-5mm}
\end{figure*}    

\begin{figure*}[t]
\centering
\begin{minipage}[t]{0.32\linewidth}
    \centering
    \includegraphics[width=0.95\linewidth]{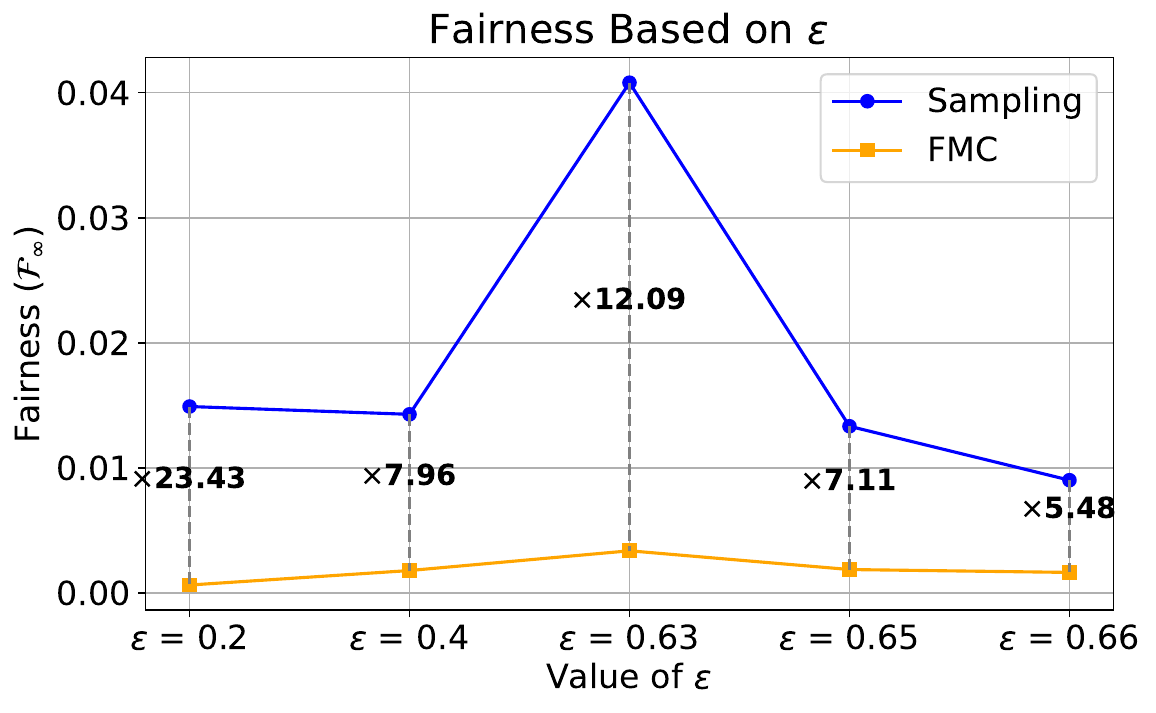}
    \vspace{-5mm}\caption{The fairness of FMC algorithm vs. standard sampling on the PopSim dataset vs. the value of $\eps$. The fairness constraint is DP.}
    \label{fig:popsim:dp}
\end{minipage}
\hfill
\begin{minipage}[t]{0.32\linewidth}
    \centering
    \includegraphics[width=0.95\linewidth]{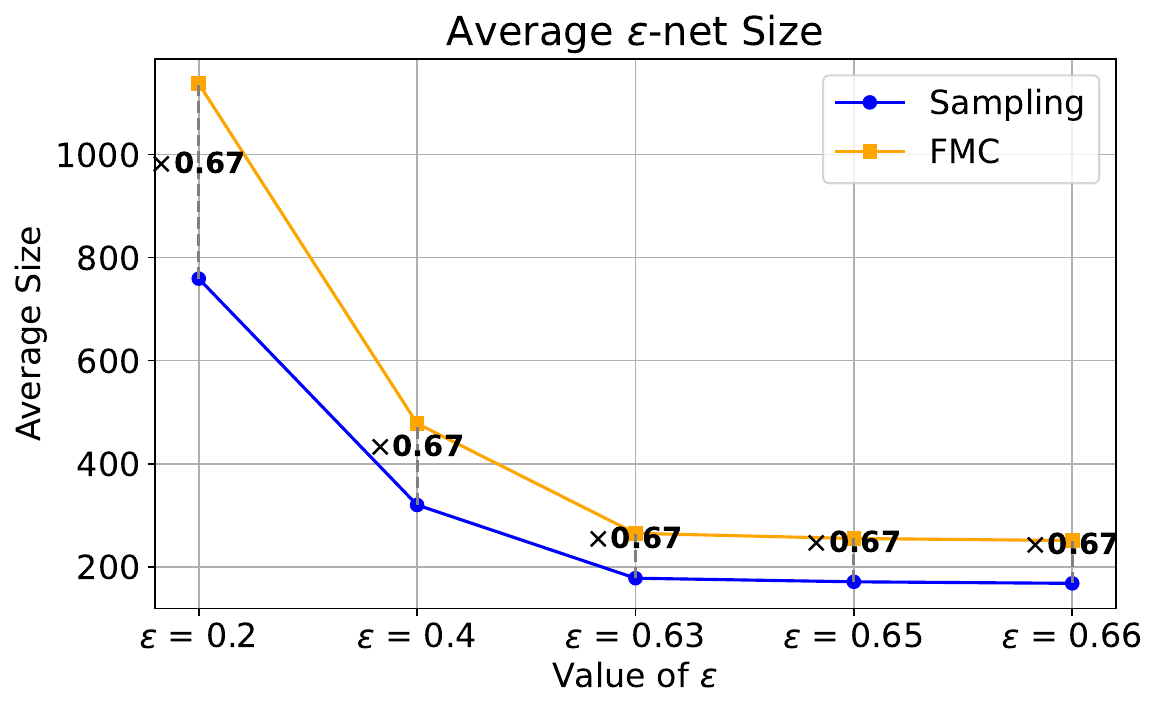}
    \vspace{-5mm}\caption{Average size of output $\eps$-net for the PopSim dataset according to DP fairness constraint.}
    \label{fig:popsim:size}
\end{minipage}
\hfill
\begin{minipage}[t]{0.32\linewidth}
    \centering
    \includegraphics[width=0.95\linewidth]{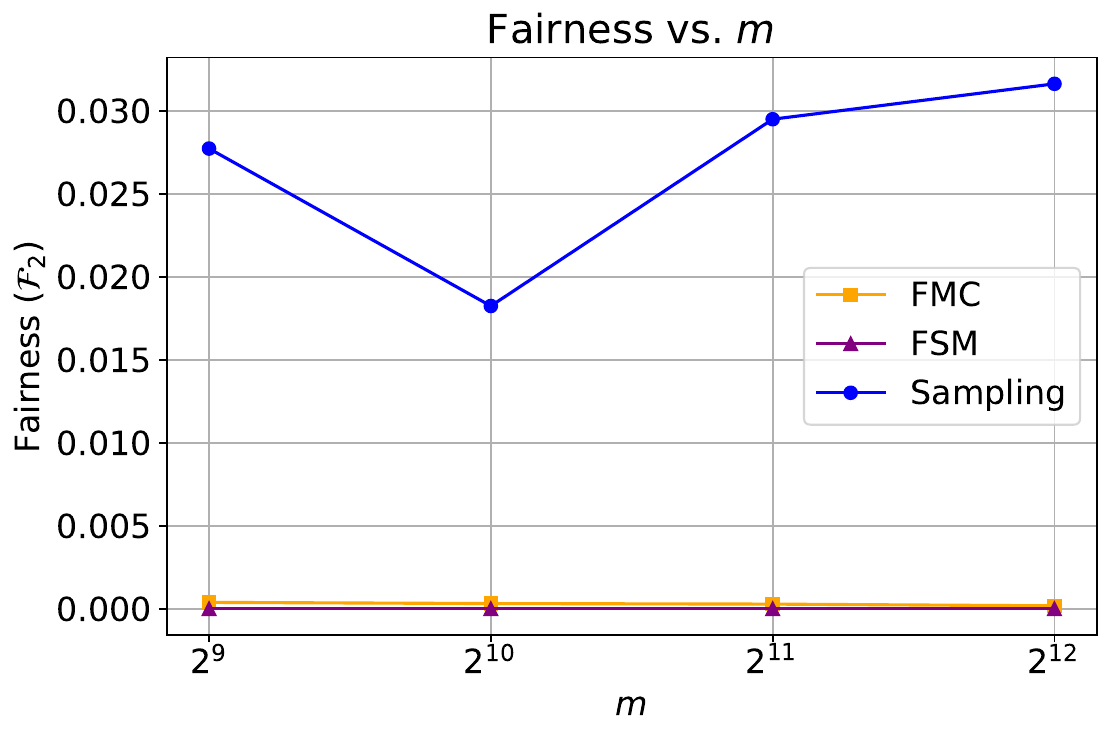}
    \vspace{-5mm}\caption{The fairness of FMC, FSM, and sampling algorithms on Synthetic Sampling task. $m$ is the number of rectangles (ranges) in 2D.}
    \label{fig:synth:epsnet:fairness}
    \vspace{-5mm}
\end{minipage}
\vspace{-4mm}
\end{figure*}

\vspace{-2.5mm}
\subsection{Experimental Setups}\label{sec:exp:setup}
We implement the following four tasks to compare the proposed algorithms with baselines:

\begin{itemize}[leftmargin=*]
    \item {\bf Database Summarization for Range Queries:}
    As illustrated in Example~\ref{ex:enet}, the goal in this setting is to find a small, representative sample of the database table. Given a collection of range queries $\mathcal{Q} = {q_1, q_2, \cdots, q_m}$, the objective is to select a small {\it fair} representative subset of tuples such that each heavy query is hit by at least one tuple from the sample.

    \item {\bf Neighborhood Hitting on Geographic Data:}
    Given a set of geographic coordinates representing individual locations, we aim to select a small {\it fair} subset of individuals (or facility points) such that every neighborhood has at least one selected point within the sample — effectively hitting all heavy spatial neighborhoods. As further discussed in \S~\ref{sec:application}, such a data summary facilitates approximate nearest neighbor search, where the goal is to return one of the top$-\eps$ closest tuples in the vicinity of a query point.

    \item {\bf Rank Regret Representative:} 
    The goal of this task is to find a small subset of tuples that hit the top-$l$ of any ranking function~\cite{asudeh2019rrr}.
    % We consider the problem of selecting a small subset of tuples from a database table such that the subset intersects the top-$l$ results for any given ranking function $f$. 
    A ranking function is defined as a linear combination of the dataset attributes. The top-$l$ tuples are determined by computing the dot product between each tuple and the weight vector $f$, followed by sorting based on the resulting scores. Our objective is to find a {\it fair} rank-regret representative of the dataset.

    \item {\bf Synthetic Data:}
    In this setup, we construct synthetic datasets with tunable parameters such as range size $m$, dataset size $n$, VC-dimension $d$, and demographic distributions. The goal is to evaluate algorithms under varying geometric complexity.
\end{itemize}

\paragraph{Algorithms:} We evaluate the following set of algorithms, as introduced in the previous sections and summarized in Table~\ref{tab:summary}:
\begin{itemize}[leftmargin=*]
    \item {\bf $\eps$-net:} As a set of baselines, for finding a standard (unfair) $\eps$-net, we consider {\it Sampling}, {\it Discrepancy}, and {\it Sketch-and-Merge (SM)} -- See \S~\ref{sec:background}\footnote{Whenever we don't mention {\it "Fair"} in the name of an algorithm, we mean the baseline unfair version. We sometimes refer to unfair algorithms as "Standard".}.  In all the reports, "{\bf Sampling}" refers to the baseline unfair $\eps$-net construction.
    \item {\bf Fair $\eps$-net:} For fair $\eps$-nets, we use {\it Fair Monte-Carlo (FMC)} algorithm and {\it Fair Sketch-and-Merge (FSM)} as introduced in \S~\ref{sec:epsnets}.
    \item {\bf Hitting Set:} For standard (unfair) hitting sets, we employ the {\it Geometric LP (GLP)} algorithm as a baseline, where LP relaxation is used to solve the geometric hitting set problem.
    \item {\bf Fair Hitting Set:} For the fair hitting set, we use the fair version, {\it Fair Geometric LP (FGLP)} as introduced in \S~\ref{sec:fghs}.
\end{itemize}

\paragraph{Measuring (un)fairness:} 
Fairness constraints are enforced based on either demographic parity (DP) or custom-ratios (CR) fairness.
We evaluate fairness in the final results by considering two measures: $\ell_2$-norm and $\ell_{\infty}$-norm. Assume $\epsnet \subseteq \points$ is the output of one of the algorithms and $\groups = \{\group_1, \group_2, \cdots, \group_k\}$ are the colors (demographic groups). The $\ell_2$-norm fairness calculates the $\ell_2$ distance of colors in $\epsnet$ versus the desired constraint vector $\mathcal{T}$:
\vspace{-0.5em}
\begin{align}
    \mathcal{F}_{2}(\epsnet, \mathcal{T}) &= \frac{1}{k} \left\| \left( \frac{|\epsnet_{\group_1}|}{|\epsnet|}, \dots, \frac{|\epsnet_{\group_k}|}{|\epsnet|} \right) - \mathcal{T} \right\|^2_2 \\
    &= \frac{1}{k} \sum_{i \leq k} \left(\frac{|\epsnet_{\group_i}|}{|\epsnet|} - \tau_i\right)^2
\end{align}

Similarly, we define $\ell_\infty$-norm fairness:
\vspace{-0.5em}
\begin{align}
    \mathcal{F}_{\infty}(\epsnet, \mathcal{T}) = \max_{i \leq k} \left| \frac{|\epsnet_{\group_i}|}{|\epsnet|} - \tau_i \right|
\end{align}

\paragraph{Datasets:}
We use four real-world datasets in our experiments. PopSim~\cite{nguyen2023popsim}, Adult~\cite{Dua:2019}, COMPAS~\cite{compas2016}, and College Admissions~\cite{Chand:2020}. A detailed description of these datasets is provided in the Appendix~\ref{app:exp:datasets}.
We also use a synthetic dataset with axis-aligned rectangles in 2D and half-spaces in higher dimensions(see Appendix~\ref{app:exp:synth} for more details).

\vspace{-3mm}
\subsection{Experiemnt Results}

\stitle{Database Summarization for Range Queries}
For this task, we use three datasets: Adult, COMPAS, and College Admission. For each dataset, we generate a collection of range queries that the output sample must hit. These ranges are constructed as hyper-rectangles in the feature space. Specifically, for each attribute, we select a threshold to partition the data into two groups, resulting in up to $2^d$ hyper-rectangles, where $d$ is the number of features. We then compute an $\eps$-net over the data points to hit these ranges for varying values of $\eps$.

Figure~\ref{fig:rds:fairness:cr} illustrates the \underline{fairness} ($\mathcal{F
}_\infty$) of the Fair Monte-Carlo (FMC) algorithm compared to the baseline (unfair) sampling approach under custom-ratio (CR) constraints. 
In the CR setting, the goal is to match a target ratio $\mathcal{T}$ of demographic groups in the output. FMC successfully enforces the desired ratio with zero unfairness, while the baseline approach fails to do so. Notably, as the target ratio for the {\em minority group} increases (towards the left side of x-axis), the standard algorithm struggles even more to satisfy the fairness constraint. In these cases, the standard sampling might not even pick any point from the minor group, which results in higher average values of $\mathcal{F}_\infty$.

The average \underline{output size} of each method is also reported in Figure~\ref{fig:rds:size}, where we can observe that the fair variants of the algorithms achieve zero unfairness while introducing a relatively small increase in the output size.

Table~\ref{tab:rds:time} reports the average \underline{running time} of both algorithms across the three datasets. The results show that the Fair Monte-Carlo (FMC) algorithm shows only a minimal overhead compared to the standard sampling approach.

Additional experimental results are presented in Appendix~\ref{app:exp:rds}, where we compare the fairness metrics of these methods across varying values of $\eps$ (Figure~\ref{fig:rds:fairness:dp}). Consistent trends in the improvement of fairness achieved by FMC can be observed. The similar results also observed on $\ell_2$-norm fairness ($\mathcal{F}_2$).

\begin{figure}
    \centering
    \includegraphics[width=0.6\linewidth]{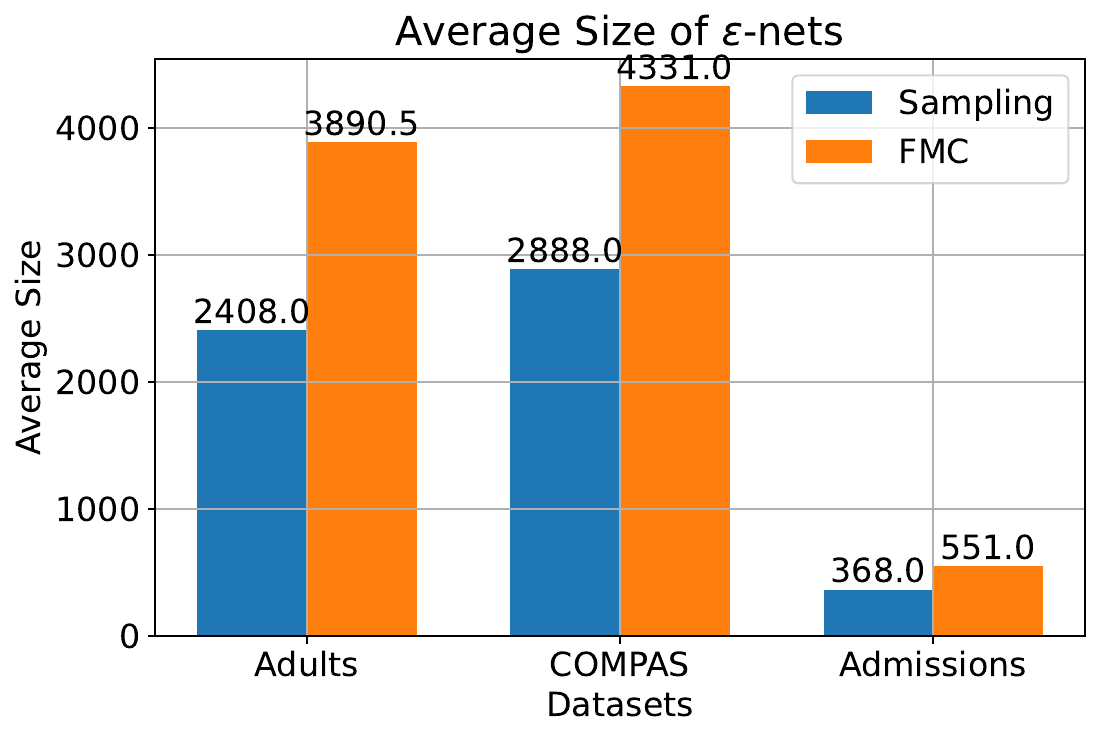}
    \vspace{-4mm}\caption{Comparing the average size of $\eps$-net produced by standard and fair algorithms for dataset summarization task.}
    \label{fig:rds:size}
    \vspace{-3mm}
\end{figure}

\renewcommand{\arraystretch}{1.2} % better row spacing

\begin{table}[tb]
\centering
\small
\begin{tabular}{@{}lrr@{}}
\toprule
\rowcolor{gray!15}
\textbf{Dataset} & \textbf{Sampling Time (s)} & \textbf{FMC Time (s)} \\
\midrule
Adults      & 2.1244 & 2.2044 (\textbf{1.03$\times$}) \\
COMPAS      & 2.4343 & 2.5470 (\textbf{1.05$\times$}) \\
Admissions  & 0.0032 & 0.0038 (\textbf{1.19$\times$}) \\
\bottomrule
\end{tabular}
\caption{Average running time of Fair Monte-Carlo and standard sampling algorithms for the dataset summarization task.}
\label{tab:rds:time}
\vspace{-6mm}
\end{table}

\begin{figure*}[t]
\centering
\begin{minipage}[t]{0.32\linewidth}
    \centering
    \includegraphics[width=0.95\linewidth]{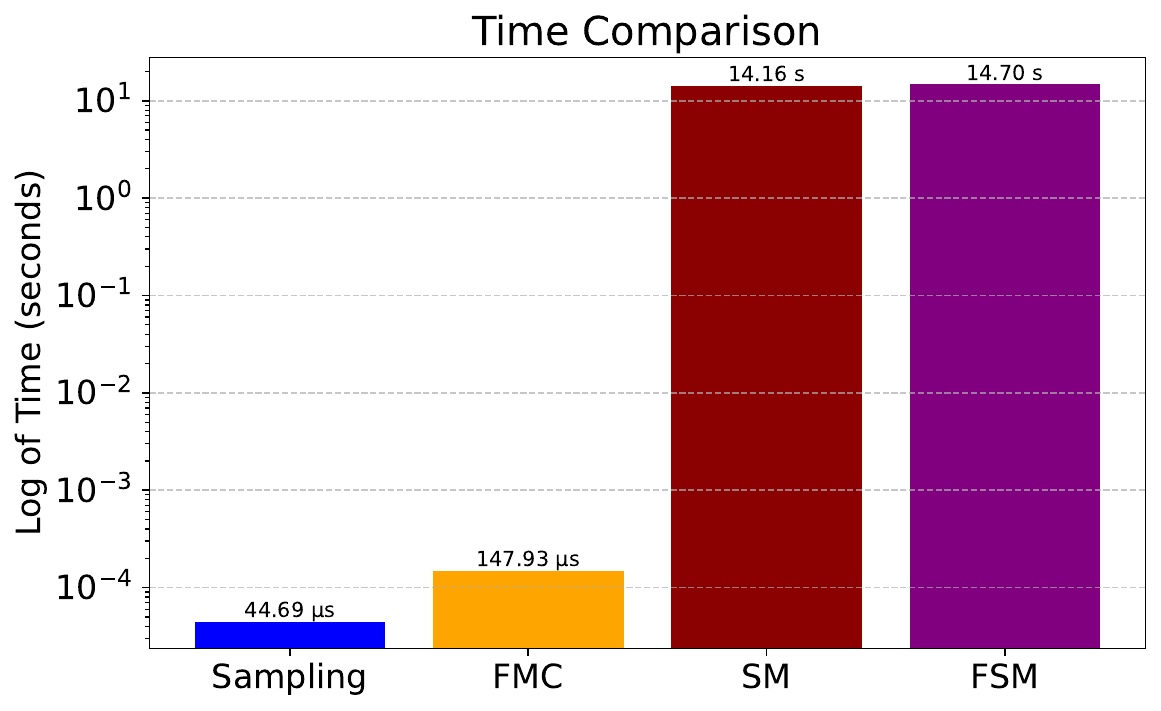}
    \vspace{-4mm}\caption{The running time of the standard methods and their fair variants. The time is averaged between different runs in the synthetic rectangle range space in 2D.}
    \label{fig:synth:time}
\end{minipage}
\hfill
\begin{minipage}[t]{0.32\linewidth}
    \centering
    \includegraphics[width=0.95\linewidth]{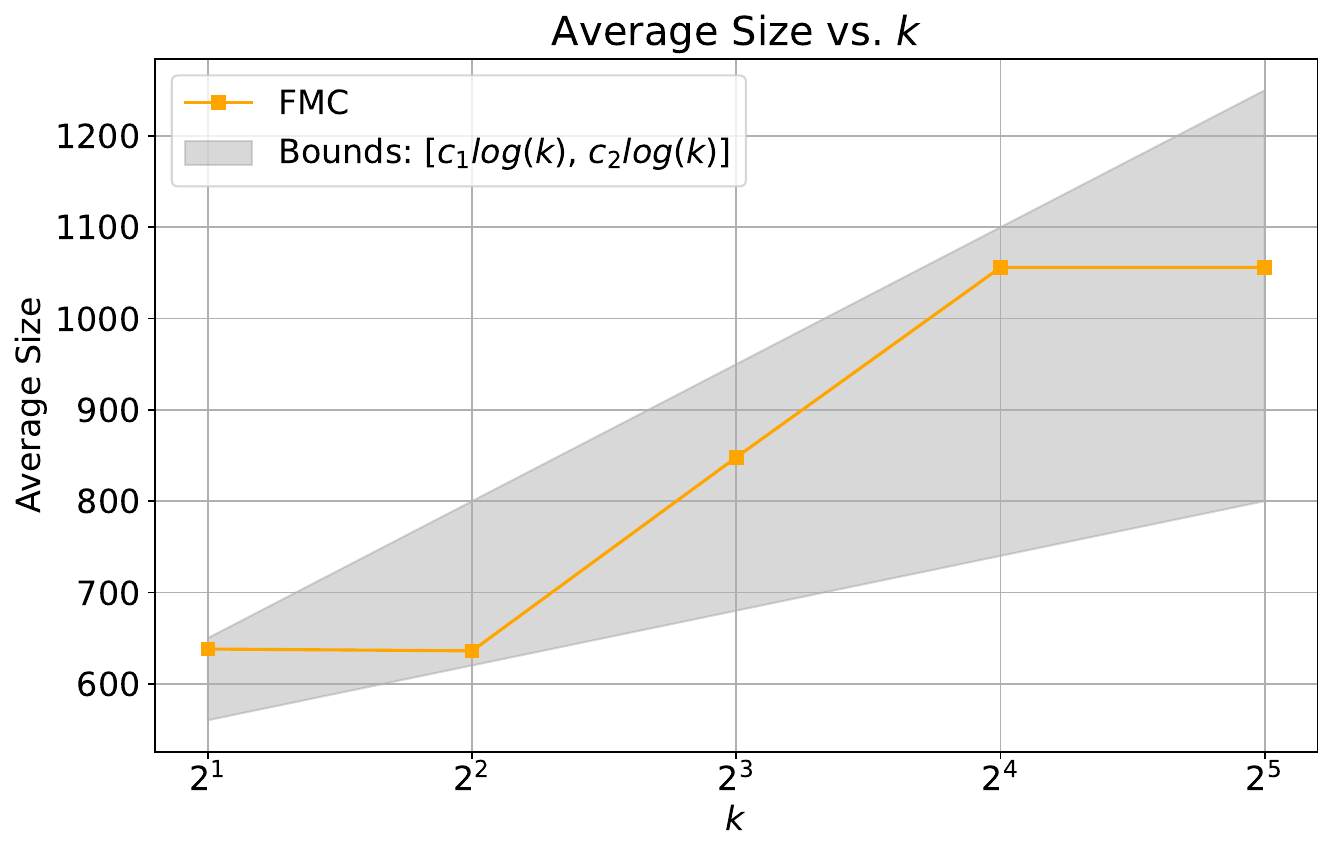}
    \vspace{-4mm}\caption{Average size of the $\eps$-net vs. the number of colors ($k$). This is the synthetic dataset with rectangle ranges with DP constraint. The gray area shows the logarithmic trend.}
    \label{fig:synth:size_k}
\end{minipage}
\hfill
\begin{minipage}[t]{0.32\linewidth}
    \centering
    \includegraphics[width=0.95\linewidth]{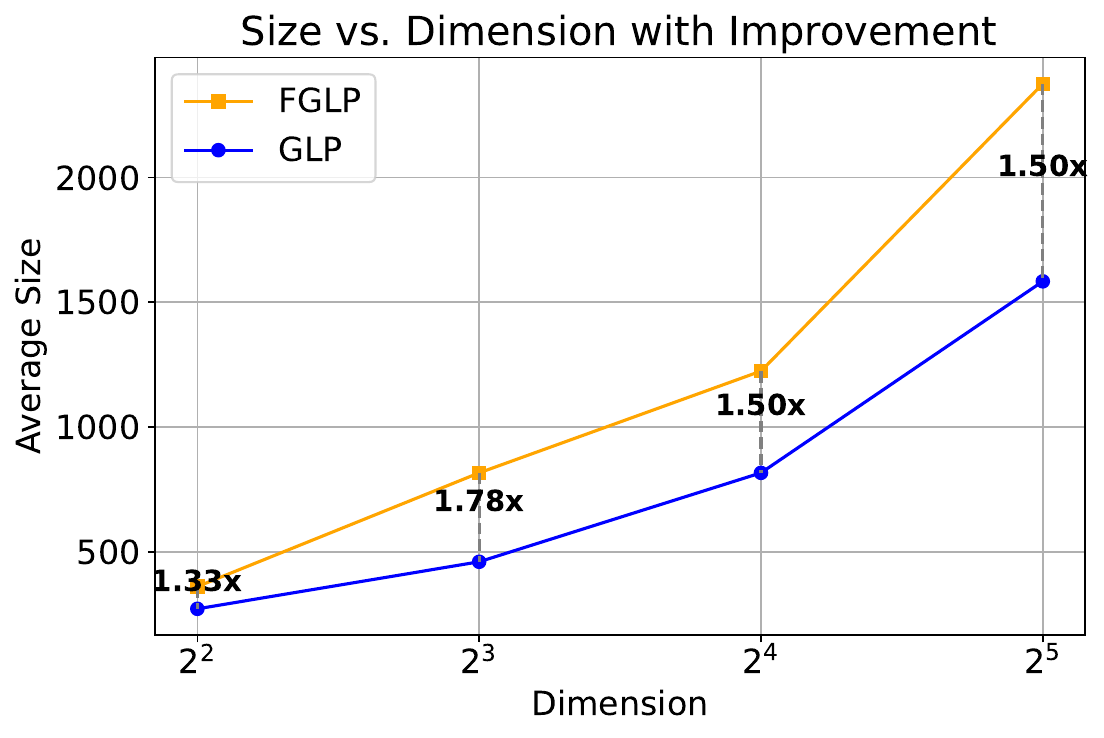}
    \vspace{-4mm}\caption{Average size of the hitting set vs. the dimension $d$. In this setting, we consider half-spaces in $\mathbb{R}^d$ space with DP constraint, and points have two colors ($k = 2$).}
    \label{fig:synth:hit-size-dim}
\end{minipage}
% \vspace{-4mm}
\end{figure*}

\stitle{Neighborhood Hitting}
For this task, we use the PopSim dataset, which contains the geographic locations of individuals. Each location is represented as a 2D point, and ranges are defined as balls of fixed radius centered around a randomly selected subset of points. The objective is to construct a fair $\eps$-net—a subset of individuals that hits all densely populated regions. We apply both the Fair Monte-Carlo (FMC) algorithm and the standard sampling approach to this setting. We also construct the Fair Hitting Set for this range space. To compute the hitting set, we use the Geometric LP (GLP) algorithm and its fairness-aware variant, Fair GLP (FGLP).

Figure~\ref{fig:popsim:dp} shows that the Fair Monte-Carlo (FMC) algorithm achieves zero \underline{unfairness} under the demographic parity (DP) constraint, whereas the standard sampling method results in a significantly unfair sample.
Figure~\ref{fig:popsim:size} further illustrates how the \underline{average size} of the output sample varies with different values of $\eps$. As $\eps$ decreases, more points are required to construct an $\eps$-net that hits all heavy ranges. In addition, the fair method introduces only a minimal increase in sample size compared to the standard approach (logarithmic to $k$). 
A comparison of the \underline{running time} of these two algorithms based on different values of $\eps$ is also provided in Appendix~\ref{app:exp:popsim} (Figure~\ref{fig:popsim:time}).

We also applied the GLP and FGLP algorithms to compute a hitting set for this range space. A comparison of their outputs is presented in Table~\ref{tab:glp-vs-fglp}. The results show that FGLP reduces unfairness 45 times, while maintaining a comparable sample size—especially considering the full dataset consists of 2M points.

\begin{table}[t]
\centering
\renewcommand{\arraystretch}{1.3}
\begin{tabular}{lccc}
\toprule
\rowcolor{gray!15}
\textbf{Algorithm} & \textbf{Output Size} & \textbf{Fairness} (\(\mathcal{F}_\infty\)) & \textbf{Runtime} (\(\mu\)s) \\
\midrule
GLP   & 336 & 0.093 & 108 \\
FGLP  & 503 & \textbf{0.002} & 403 \\
\bottomrule
\end{tabular}
\caption{GLP v.s. FGLP on PopSim dataset for Neighborhood Hitting.}
\label{tab:glp-vs-fglp}
\vspace{-9mm}
\end{table}

\stitle{Synthetic Sampling}
For this task, we generated a collection of synthetic datasets consisting of randomly sampled points and geometric ranges in the plane (rectangles) and higher dimensions (half-spaces). Details of the construction process are provided in Appendix~\ref{app:exp:synth}. We varied the dimensionality, the number of colors, and the number of ranges in these datasets.

\vspace{-2mm}
\paragraph{Results on $\eps$-nets:}     Figure~\ref{fig:synth:epsnet:fairness} compares the \underline{fairness} ($\mathcal{F}_2$) of the standard Sampling method with two fair variants: Fair Monte Carlo (FMC) and Fair Sketch-and-Merge (FSM). Both FMC and FSM achieve near-zero unfairness, whereas the standard sampling method exhibits significantly higher unfairness. Figure~\ref{fig:synth:time} presents the \underline{running time} comparison between standard Sampling and FMC, as well as between standard Sketch-and-Merge (SM) and its fair counterpart, FSM. The results show that the additional computational cost introduced by the fair methods is minimal while they significantly improve fairness.

Additional results comparing the \underline{output size} and \underline{runtime} of standard (unfair) Sampling and Discrepancy-based methods are provided in Appendix~\ref{app:exp:synthresult} (Figures~\ref{fig:synth:size} and ~\ref{fig:synth:standard}), where we compare all the three standard (unfair) algorithms.

Figure~\ref{fig:synth:size_k} shows the \underline{output size} of the Fair Monte Carlo algorithm as the {\bf number of colors} in the point set increases. The results demonstrate that an exponential increase in the number of colors leads to a linear growth in output size, aligning with the theoretical result of the $\log k$ factor.

\vspace{-2mm}
\paragraph{Results on Geometric Set Cover:} Table~\ref{tab:synth:hittingset-fair-time} presents a comparison of \underline{runtime} and \underline{fairness} between the two algorithms for constructing hitting sets. FGLP achieves zero unfairness with only a minimal increase in runtime.
Figure~\ref{fig:synth:hit-size-dim} compares the \underline{output size} of these algorithms across {\bf different dimensions}. Overall, the fair setting results in a slightly larger output size compared to the standard version, while FGLP achieves near-zero unfairness. As the dimension increases, the growth in size remains nearly constant since it depends on $\log k$ rather than the VC dimension.

\begin{table}[!tb]
\centering
\small
\renewcommand{\arraystretch}{1.5}
\begin{tabular}{@{}lcc@{}}
\toprule
\rowcolor{gray!15}
\textbf{Algorithm} & \textbf{Half-space ($d > 2$)} & \textbf{Rectangles ($d = 2$)} \\
\midrule
FGLP & \makecell{$\mathcal{F}_{\infty}$: {\bf 0.00}\\Time: 29.77s} & \makecell{$\mathcal{F}_{\infty}$: {\bf 0.00}\\Time: 1.31s} \\ \hline
GLP  & \makecell{$\mathcal{F}_{\infty}$: 0.06\\Time: 20.22s} & \makecell{$\mathcal{F}_{\infty}$: 0.28\\Time: 0.76s} \\
\bottomrule
\end{tabular}
\caption{Comparison of fairness $\mathcal{F}_{\infty}$ and average running time for fair and unfair hitting set algorithms on the synthetic dataset.}
\label{tab:synth:hittingset-fair-time}
\vspace{-7mm}
\end{table}

\stitle{Rank Regret Representatives}
To generate Rank Regret Representatives, we use three real-world datasets: Adult, COMPAS, and College Admissions. The objective is to construct an $\eps$-net that intersects the heavy top-$l$ region of any ranking function $f$ defined over the dataset features. We evaluate this across different values of $\eps = \frac{l}{n}$~\cite{asudeh2019rrr}.
The results are consistent with previous settings. Figure~\ref{fig:rrr:fairness} presents the \underline{fairness} comparison across the three datasets averaged over multiple $\eps$ values. More results are provided in Appendix~\ref{app:exp:rrr}.

\begin{figure}[!tb]
    \centering
    \includegraphics[width=0.6\linewidth]{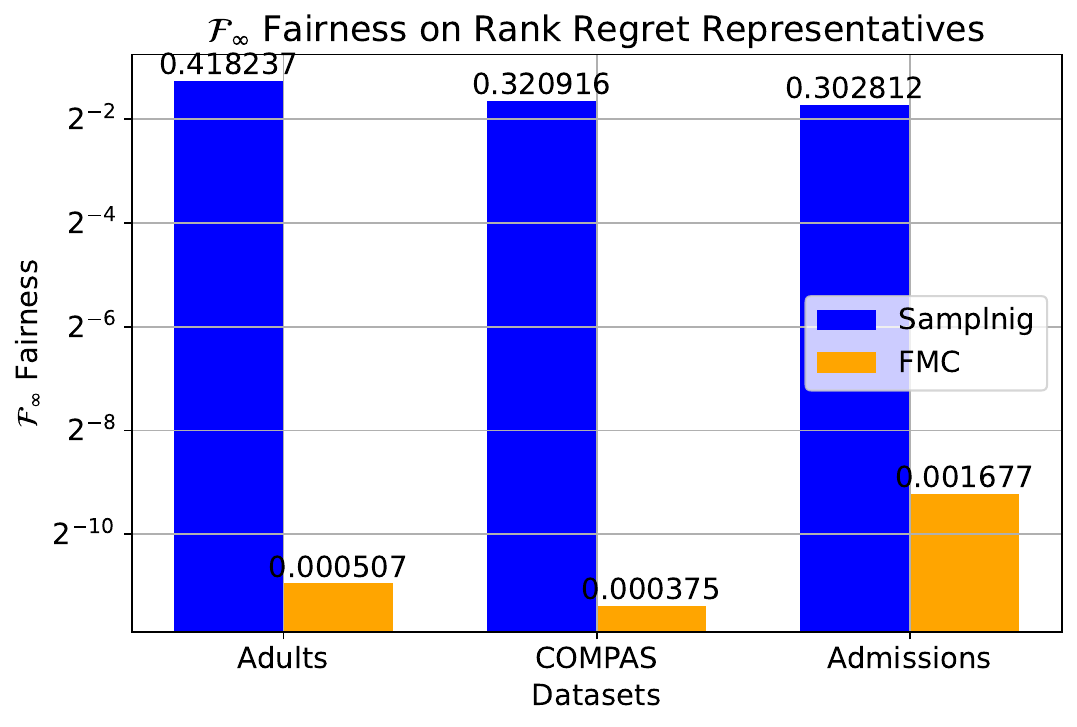}
    \vspace{-5mm}\caption{Fairness of output Rank Regret Representatives.}
    \label{fig:rrr:fairness}
    \vspace{-4mm}
\end{figure}

% \input{content/discussion}
% \vspace{-3mm}
\section{Conclusion}\label{sec:conclusion}
Motivated by their broad application in data-driven systems, in this paper, we studied the geometric approximation problems of \enet, \esample, and geometric hitting set through the lens of fairness. We formulated the problems using two notions of group fairness and proposed efficient randomized and deterministic algorithms with small approximation factors to address them. In addition to the theoretical guarantees, our experimental evaluations further demonstrated the effectiveness of our algorithms across various tasks and datasets.
\balance
% \clearpage
% \begin{acks}
% \end{acks}
\bibliographystyle{ACM-Reference-Format}
\bibliography{ref,ref-sc}

\newpage
\appendix
\section*{Appendix}
\begin{figure}[htb]
    \centering
    \includegraphics[width=0.5\linewidth]{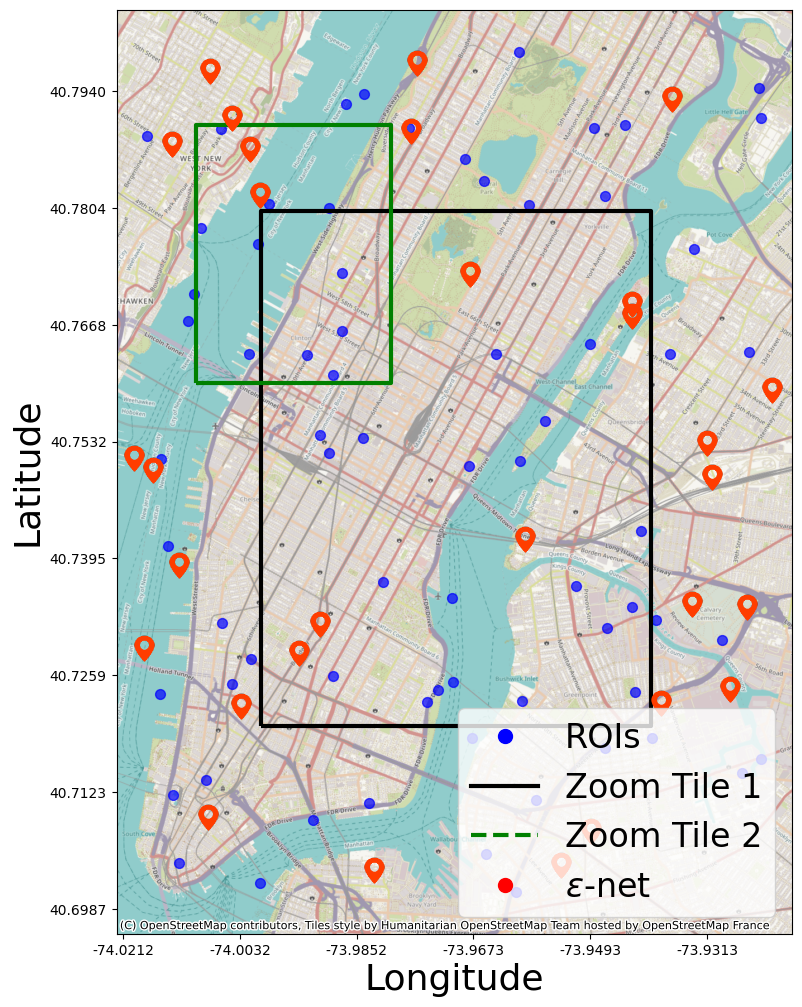}
    \caption{Application use case based on Example 1: interactive map UI with different zoom tiles. The red points are $\eps$-net samples of the whole set of regions of interest (ROI).} %\mohsen{Feel free to move this to appendix or even remove it if not necessary.}}
    \label{fig:map-example}
\end{figure}

\section{Fair $\eps$-sample}\label{app:epssample}

In this section, we discuss the Fair $\eps$-sample problem by addressing both DP and CR constraints. Most of the results are similar to Section~\ref{sec:epsnets}, so we only rely on a brief discussion.

Based on Theorem~\ref{thm:epssample} in Background, an $\eps$-sample can be built by randomly sampling 
$\gamma \geq \frac{c_0}{\eps^2}\left(d\log\frac{d}{\eps} + \log\frac{1}{\prob}\right)$ from $\points$. As a result, both the Monte-Carlo sampling-based and the discrepancy-based algorithms discussed in Section~\ref{sec:epsnets} also apply here to build a fair $\eps$-sample for satisfying DP. Specifically, 
The discrepancy-based algorithm for building $\eps$-samples follows similar steps to those in $\eps$-nets~\cite{chazelle2000discrepancy}. By constructing the Fair Matchings (as discussed in Section~\ref{sec:epsnets:disc}), one can similarly build fair $\eps$-samples.

However, it is not always possible to build an $\eps$-sample satisfying any arbitrary CR constraint. We provide a counter-example of a case where satisfying a large set of possible ratios $\mathcal{T}$ is impossible.

Assume that for a range space $\rangespace$, there is at least one range $\range^{*} \in \ranges$ that separates all the points of color $\group$ from others. In other words,
\[
    \range^* \cap \points = \points_{\group}.
\]
This is a valid assumption in practice. For example, assume the range space $\rangespace$ where $\points$ is a subset of points in $\mathbb{R}^d$ and $\ranges$ is the set of all possible half-spaces in $\mathbb{R}^d$. There can be a hyper-plane that separates at least one of the demographic groups from others with a high accuracy (this is a classification task).

Let $\epssample$ be an $\eps$-sample of this range space. Based on the Definition~\ref{def:epssample} we have:

\begin{align}
    &\left|\frac{|\epssample \cap \range^*|}{|\epssample|} - \frac{|\points \cap \range^*|}{|\points|}\right| \leq \eps\\
    &\to \left|\frac{|\epssample \cap \points_{\group}|}{|\epssample|} - \frac{|X_{\group}|}{|\points|}\right| \leq \eps
\end{align}

This means that $\frac{|\epssample \cap \points_{\group}|}{|\epssample|}$, namely the ratio of points colored as $\group$ in $\epssample$,  can not deviate more than $\eps$ from the original set. In other words, any ratio $\tau_\ell$ with $t_{\ell} \notin [(\frac{\points_{\group}}{|\points|} - \eps), (\frac{\points_{\group}}{|\points|} + \eps)]$ is not satisfiable, for $\ell\in[k]$.

In fact, even if there is a range $\range^*$ that {\bf approximately} separates one of the colors (demographic groups) from others, then, there always exists some ratios that cannot be satisfied by CR.

An example of this situation is illustrated in Figure~\ref{fig:counter-example}.

\begin{figure}[htb]
    \centering
    \includegraphics[width=0.6\linewidth]{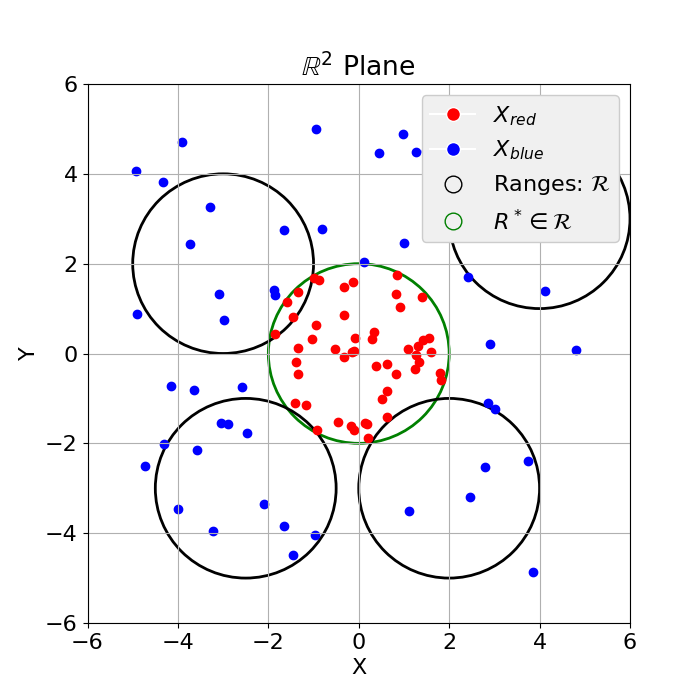}
    \vspace{-5mm}
    \caption{A counter-example for building an $\eps$-sample satisfying any CR constraints. In this example, any $\eps$-sample must contain $\frac{1}{2} \pm \eps$ fraction of red points. The reason is that there exists a range $\range^* \in \ranges$ that separates all reds from all blue points, and we have $\frac{|\points_{red}|}{|\points|} = \frac{1}{2}$.}
    \label{fig:counter-example}
    \vspace{-5mm}
\end{figure}

\section{Additional Experiments}\label{app:exp}

\subsection{Datasets Detail}\label{app:exp:datasets}
We use four real-world datasets in our experiments. PopSim\cite{nguyen2023popsim} is a large semi-synthetic dataset of 2 million individuals from Illinois, annotated with geographic and racial information. Adult\cite{Dua:2019} is a UCI benchmark from the 1994 U.S. Census, where the task is to predict income; we use gender as the sensitive attribute (66\% male, 33\% female). COMPAS\cite{compas2016} includes criminal and demographic data for over 7,000 defendants, using race as the sensitive attribute across three groups: African-American (48\%), Caucasian (37\%), and Hispanic (14\%). College Admission\cite{Chand:2020} contains application features for international students, with gender as the sensitive attribute (52\% male, 47\% female). 

PopSim contains 2,000,000 rows and three columns: latitude, longitude, and a categorical attribute representing race. Adult consists of 48,000 rows and 15 attributes, covering demographic and employment information. COMPAS includes 60,000 rows and 28 features related to criminal history and risk assessment. College Admission has 400 entries with 7 features describing applicants' academic and demographic information.

All experiments were performed on a machine with an Intel i9-12900K CPU, 64GB RAM, NVIDIA RTX 3090 GPU (24GB), running Ubuntu 22.04 and Python 3.10.

\subsection{Synthetic Dataset Construction}\label{app:exp:synth}
We constructed two classes of range spaces. In the 2D setting, the ranges were axis-aligned rectangles, and the point set consisted of randomly generated points within the unit square. Each point was randomly assigned a color from a set of $k$ colors. The values of $\eps$ were carefully chosen to produce $\eps$-nets that capture varying sizes of heavy range families, based on the selected $\eps$.

For higher-dimensional settings, we used half-spaces in $\mathbb{R}^d$ as the range family. Points were uniformly sampled from the unit hypercube, and each half-space was defined by a randomly generated normal vector and a constant. The constant was selected such that each half-space intersects at least one point from the generated point set.
The total number of points and ranges generated in all these synthetic datasets varies between $2^{12}$ to $2^{16}$.

\subsection{Additional Results on Experiments}
\subsubsection{Representative Database Sampling:}\label{app:exp:rds}
Fairness comparison under Demographic Parity (DP) and constraints for the three datasets is presented in Figures~\ref{fig:rds:fairness:dp}.

\subsubsection{Neighborhood Hitting:}\label{app:exp:popsim}
Figure~\ref{fig:popsim:time} shows the average running time of the two algorithms for finding $\eps$-net in the PopSim dataset across different values of $\eps$.

\begin{figure*}[ht]
\centering

\begin{minipage}[t]{0.45\linewidth}
    \centering
    \includegraphics[width=0.6\linewidth]{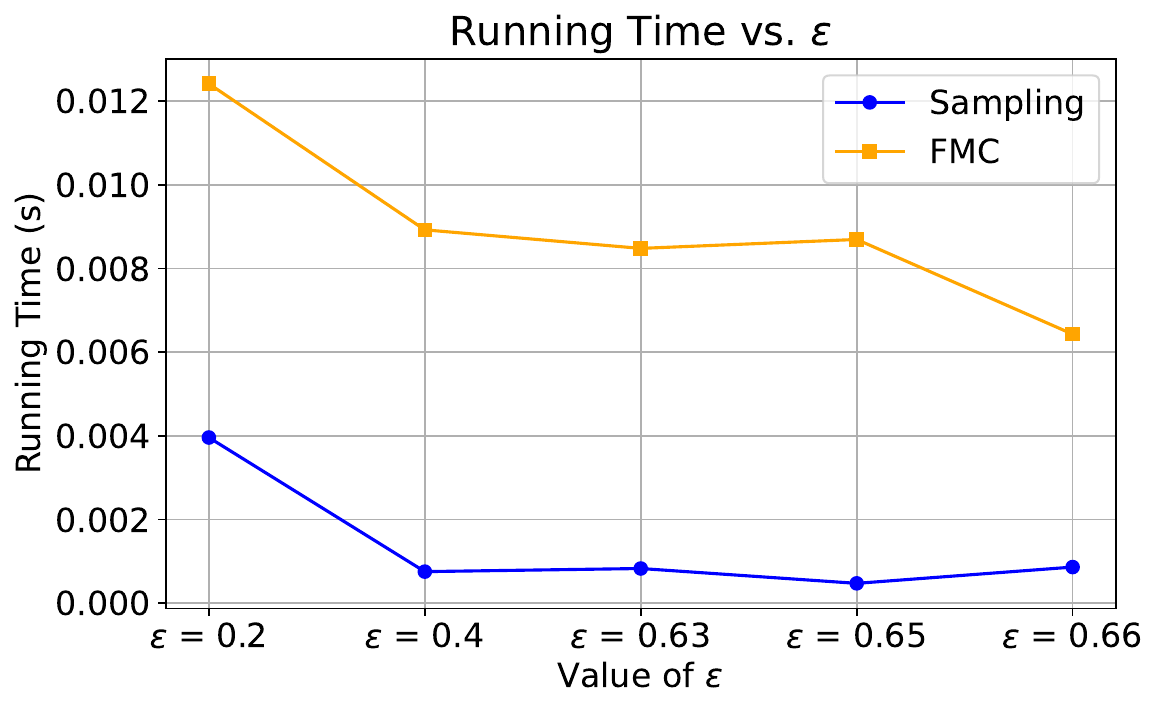}
    \caption{Running time comparison of Fair Monte-Carlo (FMC) algorithm and standard Sampling on PopSim dataset for generating $\eps$-net.}
    \label{fig:popsim:time}
\end{minipage}
\hfill
\begin{minipage}[t]{0.45\linewidth}
    \centering
    \includegraphics[width=0.6\linewidth]{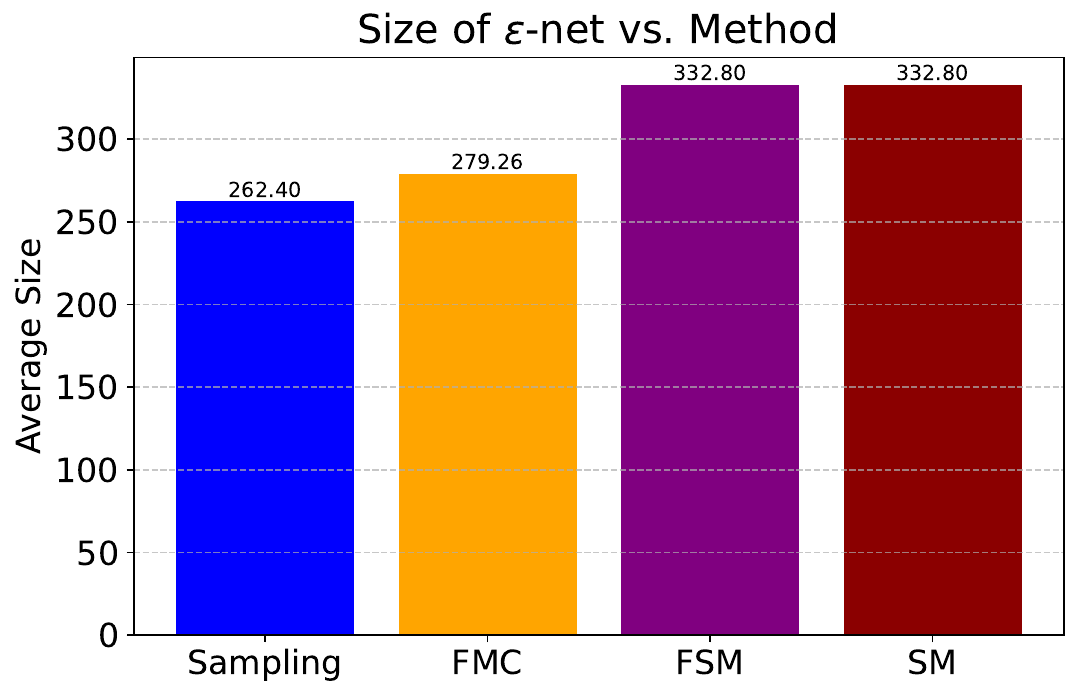}
    \caption{Comparing the size of output $\eps$-net for different methods on synthetic dataset.}
    \label{fig:synth:size}
\end{minipage}
% \hfill
% \begin{minipage}[t]{0.32\linewidth}
%     \centering
%     \includegraphics[width=0.95\linewidth]{figs/rep_db_epsnet_size_dp.pdf}
%     \caption{Comparing the average size of $\eps$-net produced by standard and fair algorithms as a representative sample of datasets.}
%     \label{fig:rds:size}
% \end{minipage}

\end{figure*}

\subsubsection{More Results of Synthetic Sampling}\label{app:exp:synthresult}
Figure~\ref{fig:synth:standard} presents a comparison of the average output size and running time of standard (unfair) algorithms for constructing $\eps$-nets. Among them, the Sketch-and-Merge (SM) algorithm outperforms the classic Discrepancy method in terms of runtime. The Sampling algorithm is the fastest overall, as it simply relies on drawing a random sample from the dataset. 

The output sizes of the Discrepancy and Sketch-and-Merge (SM) algorithms are identical. The larger output size observed for SM compared to Sampling is due to implementation constraints—specifically, SM produces outputs whose sizes are rounded up to the nearest power of two.

Figure~\ref{fig:synth:size} compares the average output size of $\eps$-nets generated by fair and unfair algorithms. The fair variants incur only a slight increase in size compared to their standard counterparts.

\subsubsection{More Results on Rank Regret Representatives}\label{app:exp:rrr}
Figure~\ref{fig:rrr:dp} shows the results of applying the Fair Monte Carlo (FMC) algorithm to compute Fair Rank Regret Representatives. Once again, FMC achieves near-zero unfairness.

\begin{figure*}[ht]
\centering
    \begin{subfigure}[t]{0.32\linewidth}
        \includegraphics[width=.95\linewidth]{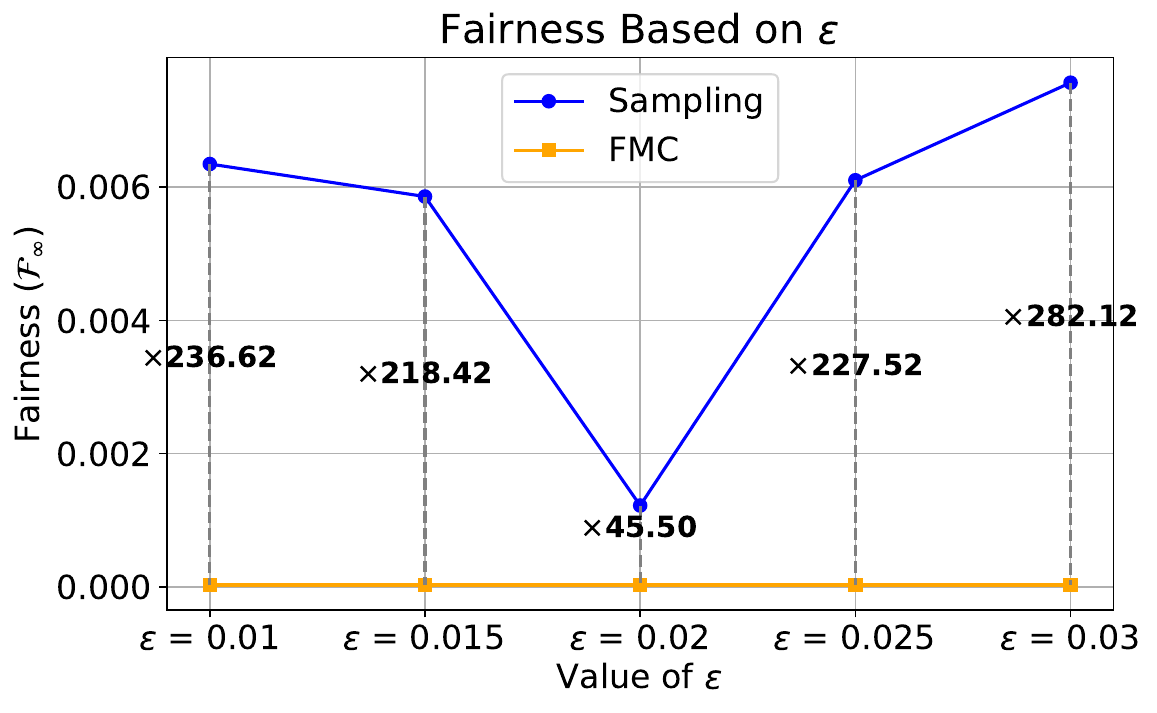}
        \caption{Adults Dataset}
        % \label{fig:cover_size_resume}
    \end{subfigure}
    \hfill
    \begin{subfigure}[t]{0.32\linewidth}
        \includegraphics[width=.95\linewidth]{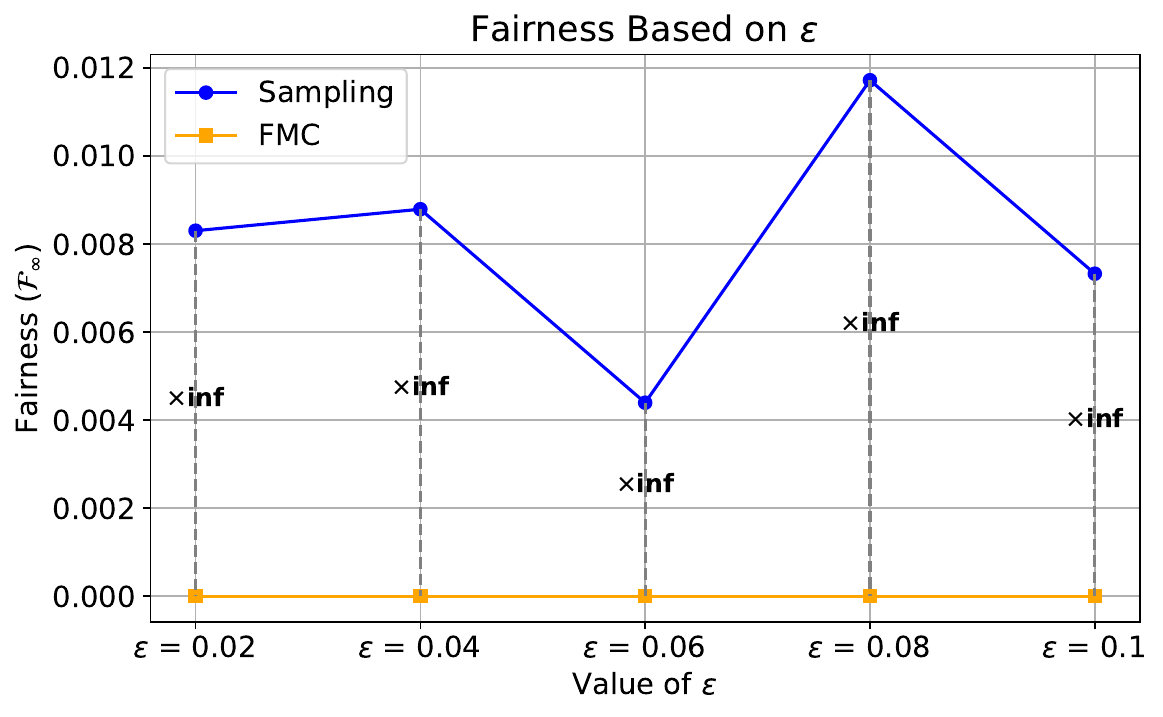}
        \caption{COMPAS Dataset}
        % \label{fig:fairness_resume}
    \end{subfigure}
    \hfill
    \begin{subfigure}[t]{0.32\linewidth}
        \includegraphics[width=.95\linewidth]{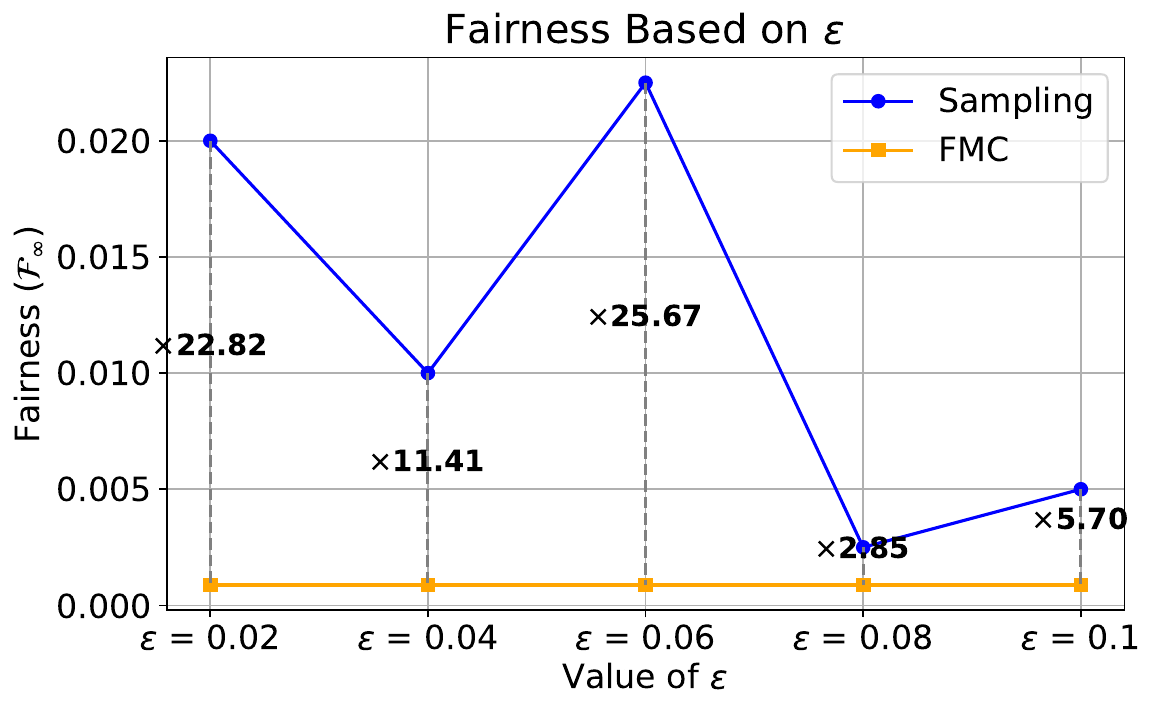}
        \caption{College Admission Dataset}
        % \label{fig:running_time_resume}
    \end{subfigure}
\caption{Measuring the fairness ($\mathcal{F}_{\infty}$) of the output $\eps$-net across different values of $\eps$ across three datasets on Representative Database Sampling task. Here, we compare the unfair Sampling and Fair Monte-Carlo sampling methods. Here, the fairness constraint is DP.}
\label{fig:rds:fairness:dp}
\end{figure*}

\begin{figure*}[ht]
\centering
    \begin{subfigure}[t]{0.45\linewidth}
        \centering
        \includegraphics[width=.6\linewidth]{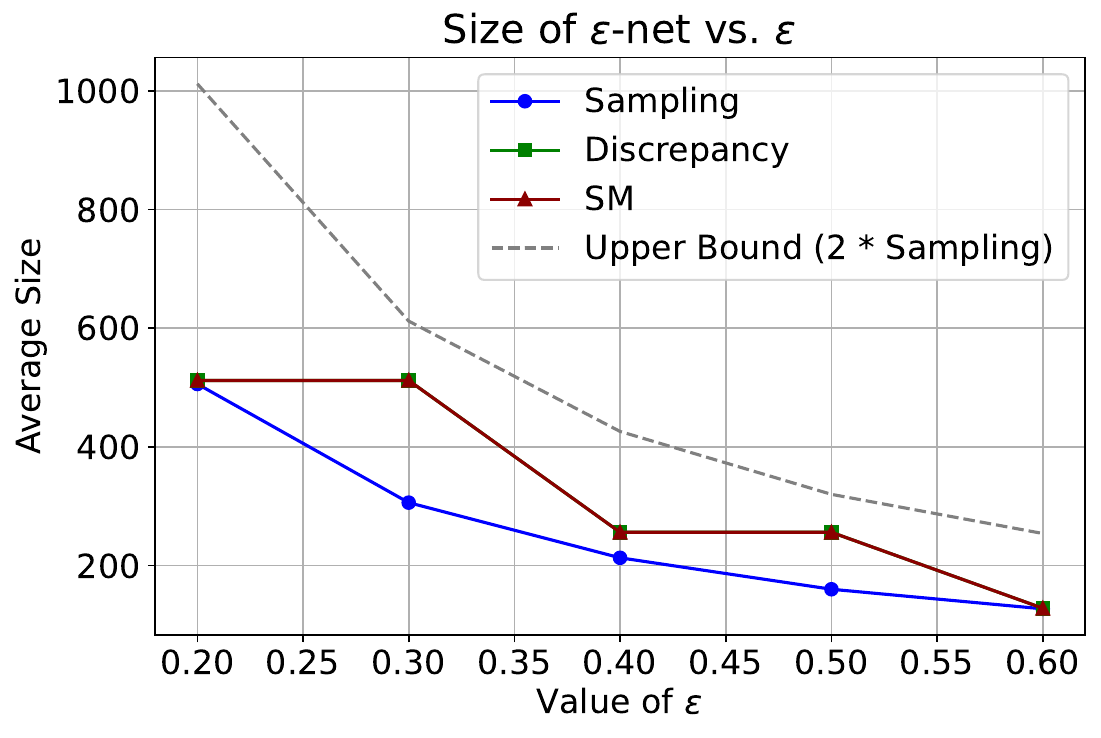}
        \caption{Average size comparison.}
        \label{fig:synth:standard:size}
    \end{subfigure}
    \hfill
    \begin{subfigure}[t]{0.45\linewidth}
        \centering
        \includegraphics[width=.6\linewidth]{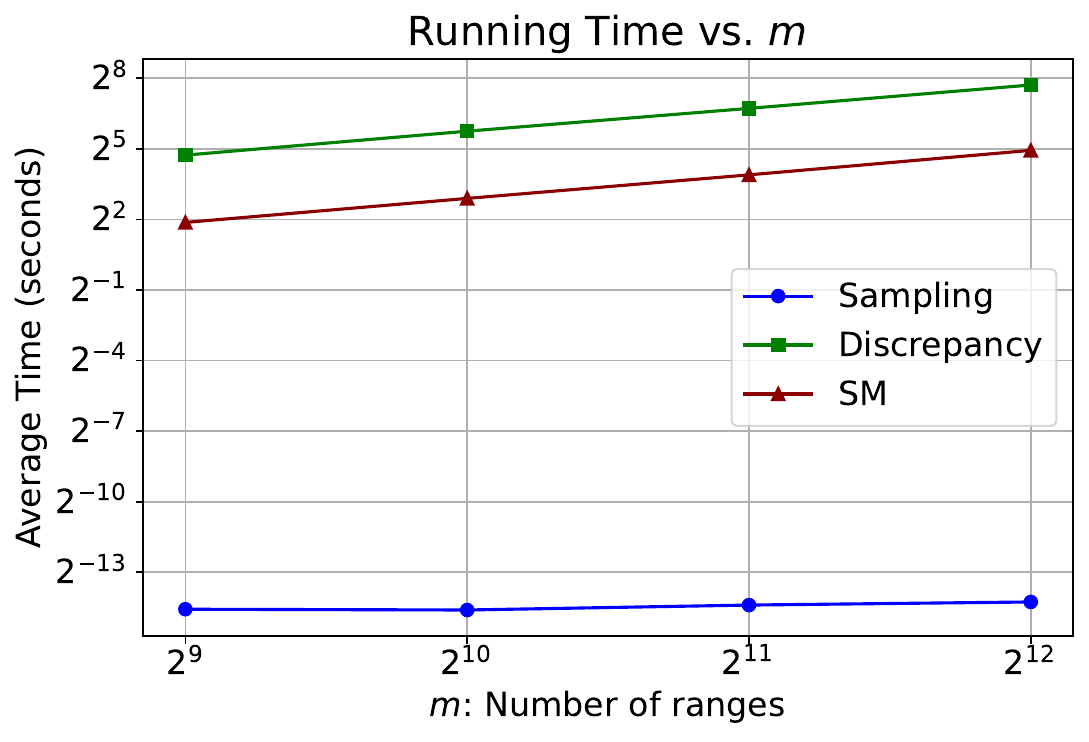}
        \caption{Running time comparison.}
        \label{fig:synth:standard:time}
    \end{subfigure}
\caption{The comparison between the standard algorithms for building $\eps$-net. The Sketch-and-Merge (SM) algorithm achieves a better runtime compared to the classic Discrepancy algorithm, and Sampling is the fastest. This experiment is run on the synthetic dataset.}
\label{fig:synth:standard}
\end{figure*}

\begin{figure*}[ht]
\centering
    \begin{subfigure}[t]{0.32\linewidth}
        \includegraphics[width=.95\linewidth]{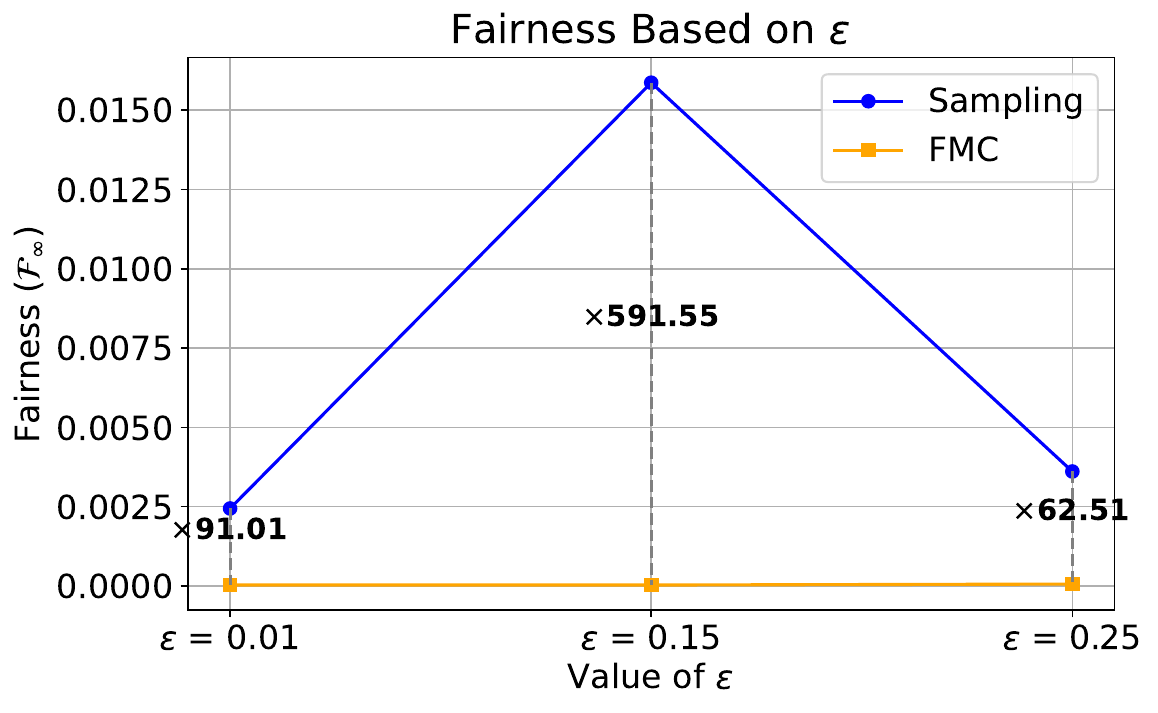}
        \caption{Adults Dataset}
        % \label{fig:cover_size_resume}
    \end{subfigure}
    \hfill
    \begin{subfigure}[t]{0.32\linewidth}
        \includegraphics[width=.95\linewidth]{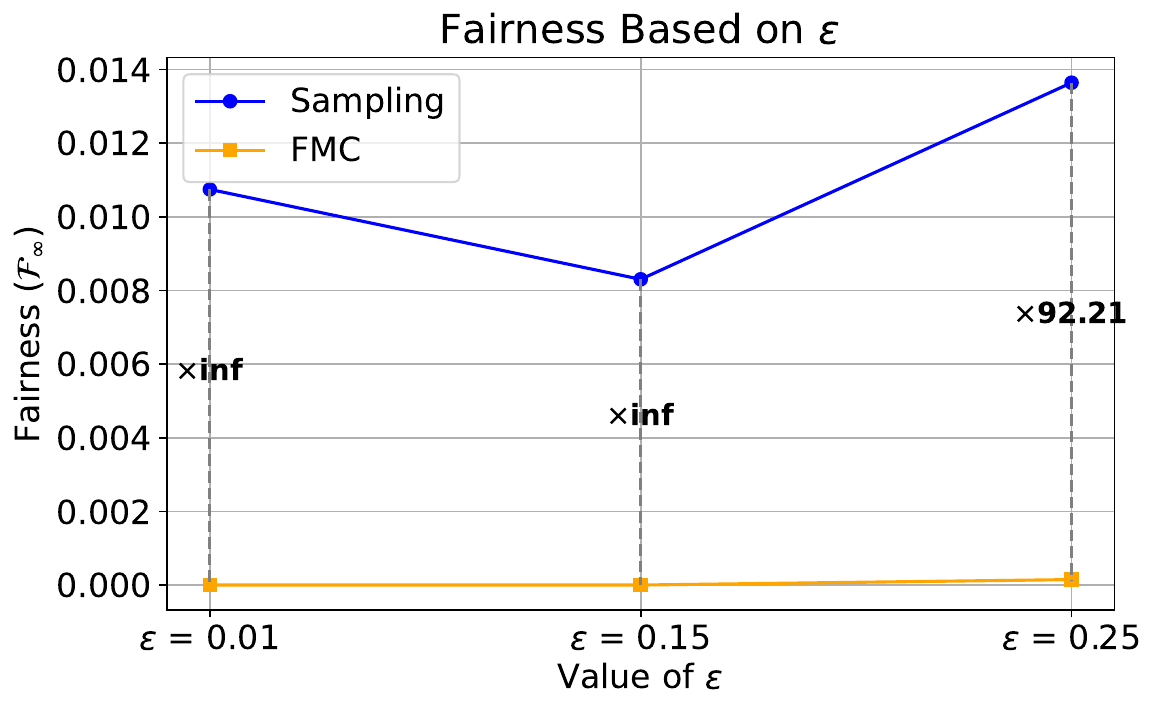}
        \caption{COMPAS Dataset}
        % \label{fig:fairness_resume}
    \end{subfigure}
    \hfill
    \begin{subfigure}[t]{0.32\linewidth}
        \includegraphics[width=.95\linewidth]{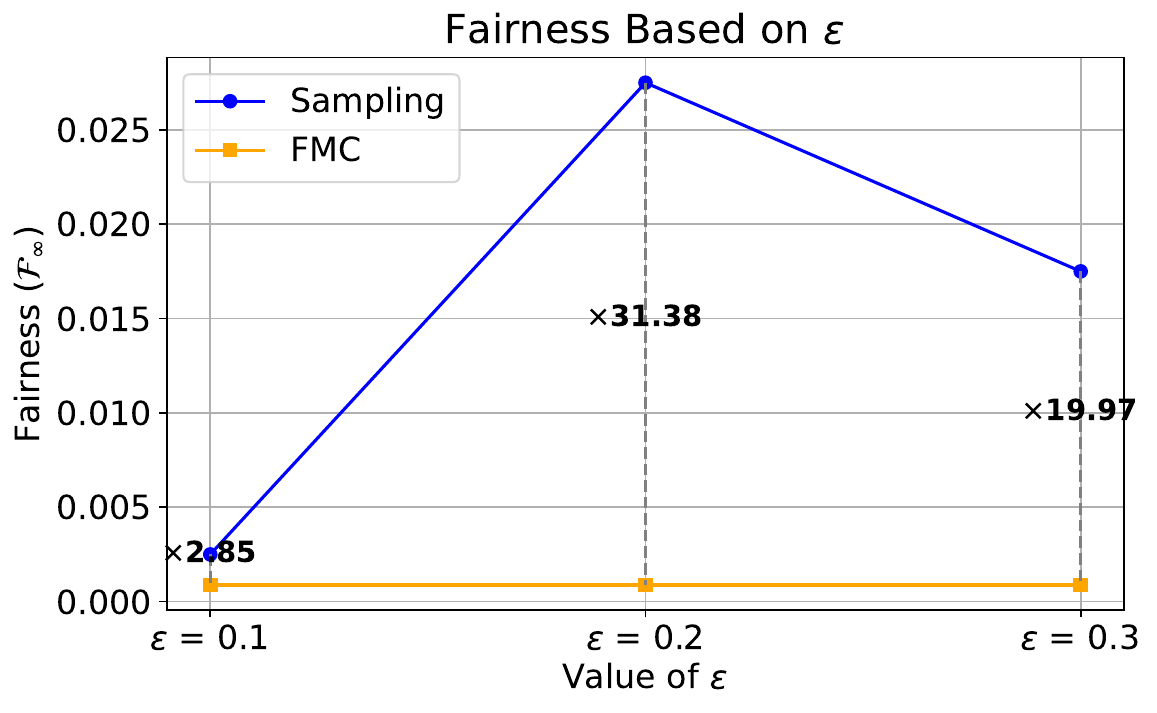}
        \caption{College Admission Dataset}
        % \label{fig:running_time_resume}
    \end{subfigure}
\caption{The result of applying fair algorithms in finding the fair Rank Regret Representatives on the three datasets.}
\label{fig:rrr:dp}
\end{figure*}
% \clearpage
% \input{content/revision}
% \include{content/oldversion}

% \appendix

\end{document}